\documentclass[amsmath,11pt,amssymb]{revtex4}
\usepackage{graphicx}
\usepackage{dcolumn}
\usepackage{bm}
\usepackage{times}
\usepackage{soul}
\usepackage{url}
\usepackage[hidelinks]{hyperref}
\usepackage[utf8]{inputenc}
\usepackage[small]{caption}
\usepackage{graphicx}
\usepackage{amsmath}
\usepackage{amsthm}
\usepackage{booktabs}
\usepackage{algorithm}
\usepackage{algorithmicx}
\usepackage{algpseudocode}
\usepackage[switch]{lineno}
\usepackage{float}
\usepackage{subcaption}
\usepackage{wrapfig}
\usepackage{graphicx}
\usepackage{paralist}
\usepackage{enumitem}
\usepackage{multirow}
\usepackage{mathtools}
\usepackage{amsfonts}
\usepackage{nicefrac}
\usepackage{color}
\usepackage{dcolumn}
\usepackage{soul} 
\usepackage{comment}
\theoremstyle{definition}
\newcommand{\remove}[1]{}
\def\math#1{$#1$}
\newcommand{\param}{\boldsymbol{\theta}}
\newcommand{\obj}{\mathcal{E}}
\newcommand{\degree}{\boldsymbol{\delta}}
\newcommand{\Steady}{ {\mathbf{x}}}
\newcommand{\steady}{ {x}}

\newtheorem{theorem}{Theorem}
\newtheorem{definition}{Definition}
\newtheorem{lemma}{Lemma}
\newtheorem{corollary}{Corollary}

\begin{document}
\preprint{APS/123-QED}

\title{Efficient parameter inference in networked dynamical systems via steady states: A surrogate objective function approach integrating mean-field and nonlinear least squares}

\author{Yanna Ding}  
\author{Malik Magdon-Ismail}
\author{Jianxi Gao}
\email{gaoj8@rpi.edu}
\affiliation{Department of Computer Science, Rensselaer Polytechnic Institute, Troy, New York 12180, USA}

\date{\today}%

\begin{abstract} 
In networked dynamical systems, inferring governing parameters is crucial for predicting nodal dynamics, such as gene expression levels, species abundance, or population density. While many parameter estimation techniques rely on time-series data, particularly systems that converge over extreme time ranges, only noisy steady-state data is available, requiring a new approach to infer dynamical parameters from noisy observations of steady states. However, the traditional optimization process is computationally demanding, requiring repeated simulation of coupled ordinary differential equations (ODEs). To overcome these limitations, we introduce a surrogate objective function that leverages decoupled equations to compute steady states, significantly reducing computational complexity. Furthermore, by optimizing the surrogate objective function, we obtain steady states that more accurately approximate the ground truth than noisy observations and predict future equilibria when topology changes. We empirically demonstrate the effectiveness of the proposed method across ecological, gene regulatory, and epidemic networks. Our approach provides an efficient and effective way to estimate parameters from steady-state data and has the potential to improve predictions in networked dynamical systems.

\end{abstract}
 
\maketitle

\section{\label{sec:intro}Introduction}

Knowing a network's dynamics is essential for understanding and predicting its behavior in complex systems, with broad application in ecology~\cite{may1974biological,bascompte2013mutualistic,holling1973resilience,scheffer2001catastrophic,holland2002population,landi2018complexity}, biology~\cite{barabasi2004network,dai2012generic,casey2006piecewise,karlebach2008modelling}, epidemiology~\cite{liu1987dynamical,pastor2015epidemic}, energy supply~\cite{kessel1986estimating,zimmerman2010matpower}, neurophysiology~\cite{cabral2014exploring}, and organizational psychology~\cite{wang2016dynamic}.  
Network dynamics often consists of three fundamental components: dynamical models, parameters, and network topology. The dynamical models are closed-form mathematical expressions, including polynomials, trigonometric functions, etc. Dynamical parameters are essential numerical values, such as polynomial orders and scaling factors within these mathematical expressions. The network topology determines the  coupling strengths between nodes. Our focus is parameter estimation, given underlying topology, dynamical models, and observed steady states. Accurate parameter estimation is critical for computing important network tasks, including computing tipping point leading to systems collapses~\cite{pereira2015control,jiang2018predicting,morone2019k,phillips2020spatial,macy2021polarization,zhang2022estimating,liu2022network,wu2023rigorous}, network control~\cite{liu2011controllability,gao2014target}, synchronization~\cite{arenas2008synchronization}, and many more.

There are existing methods for inferring parameters~\cite{gugushvili2012sqrt,ramsay2007parameter}, topology~\cite{prasse2022predicting}, and dynamical models~\cite{bongard2007automated,schmidt2009distilling,brunton2016discovering,gao2022autonomous}. These approaches typically use a time series of system states, which is much richer information than just steady states.  
For inferring dynamical models, the two approaches are to use semi-parametric approaches like~\cite{bongard2007automated,schmidt2009distilling,brunton2016discovering,gao2022autonomous} or some deep learning methods~\cite{chen2018neural,zang2020neural,huang2021coupled}. Notably, these deep learning methods often require extensive time-series data for effective learning and may struggle to generalize to unobserved dynamical systems.
Also note that existing methods relying on a time series of states can be costly or even impossible in certain scenarios~\cite{kramer2014hamiltonian}. Ecological systems evolve over millennia~\cite{slobodkin1980growth,oro2022long}, and we only observe the current steady state of each species. For biomedical networks, equilibration is at small time scales, making it difficult for experimental devices to capture interim states~\cite{karlebach2008modelling,boiger2019continuous}. Despite recent progress to learn microscopic dynamics from multiple sets of steady states generated by perturbation experiments~\cite{kramer2014hamiltonian,barzel2015constructing,santra2018fitting}, there is still a lack of a universal framework to infer dynamical parameters from a single noisy steady-state observation of a networked dynamics. 

Parameter estimation for differential equations is often time-consuming due to repeated numerical integration. One approach to address this issue is fitting proposed models to the derivatives of the observed time-series data approximated using finite difference~\cite{prasse2022predicting}. However, this approach may yield inaccurate results when the data resolution is low, and, of course, it is not applicable when only the equilibrium state is available. One common objective function to estimate parameters is to minimize the discrepancy between observed states and those generated by the learned model using numerical integration techniques~\cite{ramsay2007parameter}. However, repeated numerical integration to get network states is computationally intensive, especially for coupled complex systems with nonlinear dynamics.
In this paper, our aim is to improve the efficiency of parameter estimation for ODEs when the only available data is noisy observation  of steady states. 

The mean-field network reduction approach, introduced in \cite{gao2016universal}, simplifies coupled systems while attempting to preserve their original dynamics. This approach collapses a high-dimensional network into a one-dimensional effective state which drives the dynamics of all other states. Extensive studies have been conducted to evaluate the accuracy of this approach \cite{gao2016universal,kundu2022accuracy}. Alternative  dimensionality reduction methods have been explored, including those proposed in~\cite{ laurence2019spectral,tu2021dimensionality,vegue2023dimension}. Furthermore, researchers have extended network reduction approaches to forecast network collapse~\cite{jiang2018predicting}, predict steady states from incomplete network topology \cite{jiang2020true} and infer degrees from observed steady states \cite{jiang2020inferring}. While these previous works assume prior knowledge of dynamical parameters, our work builds upon the mean-field approach and combines it with nonlinear least squares (NLS) to estimate unknown parameters.
 
In this work, we propose a surrogate objective function based on mean-field estimation of steady states that significantly reduces computational time while preserving the shape of the exact NLS objective function. The parameter-specific steady states are approximated using a low-dimensional, decoupled system of ODEs in the mean-field approximation. We demonstrate empirically that NLS inference using the surrogate objective produces parameters that recover steady states that match the ground truth more accurately than the observed noisy steady states. Furthermore, we demonstrate that the parameters inferred through our method can predict state evolution when the network undergoes changes.

\begin{figure*}[t]
 \centering
\includegraphics[width=\textwidth]{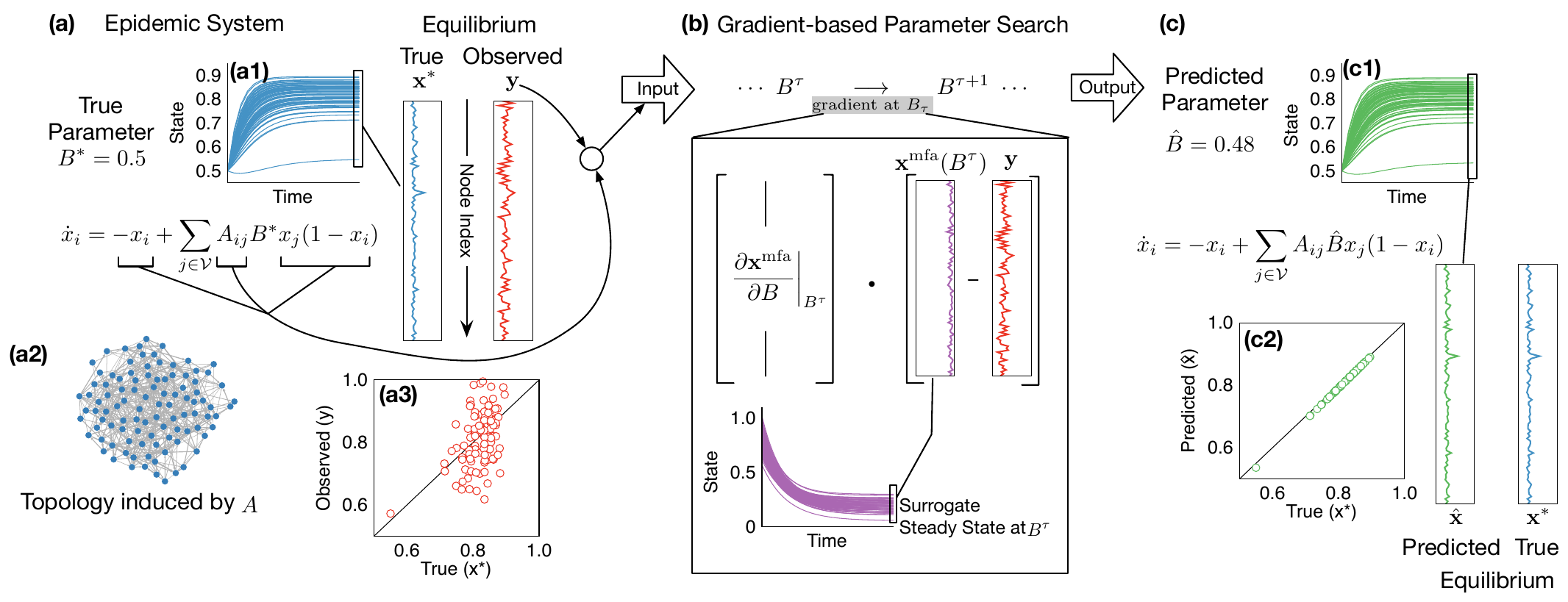} 
\caption{\label{fig:problem_setup}
Workflow Illustration on Epidemic Dynamical Systems.
(a) Ground Truth Epidemic System and Observed Equilibrium: 
(a1) Ground truth steady states ($\mathbf{x}^*$) simulated using coupled ODEs with the ground truth parameter $B^*=0.5$.
(a2) Network topology induced by the adjacency matrix $A$, generated using the Erd\H{o}s R\'enyi model ($N=100$, average degree $12$).
(a3) Comparison of observed states ($\mathbf{y}$) with true steady states. Observations are sampled from a normal distribution with mean $x_i^*$ and standard deviation $10\%\lvert x_i^* \rvert$.  
(b)  Optimization of Surrogate Objective Function (Eq.~(\ref{eq:obj1})): 
The optimization process takes as input mathematical expressions representing the self-dynamics and interaction terms, the degree of each node, and the observed steady states $\mathbf{y}$. A gradient-based parameter search algorithm is employed for optimization. Note that numerical integration is performed for each query parameter $B^{\tau}$ during the optimization. The detailed computation of the gradients is provided in Appendix~\ref{sec:grad}. 
(c) Inferred Parameter $\hat{B}$ and Predicted Equilibrium: 
(c1)  State evolution of the ODE with the inferred parameter $\hat{B}$, resulting in the predicted equilibrium $\hat{\mathbf{x}}=\mathbf{x}(\hat{B}) $.
(c2) Close alignment between predicted steady states and ground truth.
} 
\end{figure*}

\section{\label{sec:surrogateODE}Parameter Estimation via a Surrogate Objective Function}

\subsection{Problem Setup}

We model a networked dynamical system using the framework in~\cite{barzel2013universality}. The networked system is a graph \math{\mathcal{G}}, with nodes $\{1,\ldots,N\}$ and adjacency matrix $A$, where $A_{ij}=1$ if $(i,j)$ is an edge in~\math{\mathcal{G}}. We denote the in-degrees as $\degree^{\text{in}} = A\mathbf{1}  $ and out-degrees as  $\degree^{\text{out}}=\mathbf{1}^{\top}A$, where $\mathbf{1}$ is a vector consisting of all ones. For undirected networks, we have $\degree^{\text{in}}=\degree^{\text{out}}=\degree$. Each node \math{i} has a time dependent scalar state \math{x_i(t)}. The states follow a general form of coupled ODE,
\begin{align}
\dot{x}_{i}=F_{i}(\mathbf{x},\param,A)=
f(x_{i},\boldsymbol{\theta}) + \sum_{j=1}^{N} A_{ij}g(x_{i},x_{j},\boldsymbol{\theta}).\label{eq:ode}
\end{align}
The functions $f$ and \math{g} give the intrinsic and interaction forces. The interaction is modulated by $A_{ij}$. The functions \math{f,g} and the graph topology induced by the adjacency matrix \math{A} are given, while the ODE parameters \math{\boldsymbol{\theta}} are unknown and must be learned. We study cooperative dynamics,  where each neighbor's interaction positively affects each node's survival~\cite{wu2023rigorous} ($\frac{\partial g}{\partial x_j} \geq 0$). Since $A_{ij}\geq 0$, nodes with more neighbors have higher equilibria. Table~\ref{tab:dyn} summarizes the three specific dynamics we use in our experiments. 

One standard approach to learn the parameters is iterative NLS. Let $\steady_i(\param)$ be the steady states of the system for parameter $\param$. NLS finds parameters \math{\boldsymbol{\theta}} that minimize the error
\begin{align}
\obj(\boldsymbol{\theta},\mathbf{y}) = \frac{1}{N} 
\lVert\Steady(\boldsymbol{\theta}) - \mathbf{y}\rVert^{2} =\frac1N\sum_{i=1}^N(\steady_i(\param)-y_i)^2. \label{eq:obj}
\end{align}
where $\mathbf{y}\in\mathbb{R}^{N}$ denote the observed steady states.
(Boldface is used for vectors and \math{\lVert\cdot\rVert} is the Euclidean norm.)

The optimization process employs a gradient-based algorithm, such as Conjugate Gradient Descent, to traverse the parameter space. This algorithm iteratively evaluates the objective function, moving to the next query parameter that reduces the objective function. At each query parameter, the objective function and its gradient are computed. This requires solving the ODEs to obtain the steady states corresponding to each parameter value tested. Our main contribution is an efficient method to compute an approximation to steady states at each query parameter, together with an empirical study of the effectiveness of our approach.

The problem setup is illustrated in Fig.~\ref{fig:problem_setup}. The system evolves under the ground truth parameter $\param^{*}$ to produce a time series of ground truth states ending in the ground truth steady states \math{x_i^{*}} (Fig.~\ref{fig:problem_setup}(b)). One only observes noisy estimates of the steady states, $y_{i} = x_{i}^{*} + \epsilon_{i} $, where $\epsilon_{i}$ are noise. Our goal is to use the observed noisy steady states \math{y_i} to learn a parameter  that reproduces the true steady states ${x}_i^{*}$. To summarize, our inputs are the network's topology $A$, the dynamical functions $f$ and $g$, and the observed steady states $\mathbf{y}$. The output is a parameter $\hat{\param}$ that minimizes our defined objective function $\mathcal{E}(\param,\mathbf{y})$. We aim for the estimated parameter $\hat{\param}$ to approximate the true parameter $\param^*$ and for the predicted steady states   $\hat{\Steady}=\mathbf{x}(\hat{\param})$ to closely align with the ground truth $\Steady^*$, surpassing the observational data's proximity, i.e., $\hat{\Steady}\approx \mathbf{x}^*$ and $\lVert\hat{\Steady} - \mathbf{x}^* \rVert\leq \lVert \mathbf{y} - \mathbf{x}^*\rVert$.

\begin{table*}
\caption{\label{tab:dyn}Summary of the three dynamics analyzed in this paper. The ecological dynamics models a plant network projected from plant-pollinator mutualistic relationships. $B$ represents the incoming migration rate. The second term describes the logistic growth with carrying capacity $K$ and Allee constant $C$. The third term stands for mutualistic interaction that saturate for large $x_i,x_j$~\cite{gao2016universal,kundu2022accuracy}. The gene regulatory dynamics is adapted from the Michaelis-Menten model~\cite{alon2006introduction,karlebach2008modelling}. The first term denotes degradation ($f=1$) or dimerization ($f=2$). The Hill coefficient $h$ represents the level of cooperation of the gene regulation~\cite{gao2016universal,wu2023rigorous}. In the epidemic process, node $i$ is susceptible ($0\leq x_i < 1$) or infected ($x_i=1$). Infected nodes spread the pathogen to their neighbors at rate $B>0$~\cite{wu2023rigorous}. 
}
\begin{ruledtabular}
\begin{tabular}{lll}
Dynamics & State &Formula   \\
\colrule
Ecology  & Biomass &$\dot{x}_{i} = B + x_{i}\left(1-\frac{x_{i}}{K}\right)\left(\frac{x_{i}}{C}-1\right) + \sum_{i=1}^{N} A_{ij} \frac{x_{i}x_{j}}{D + E x_{i} + H x_{j}}$\\  
Gene regulatory &Gene expression level &$\dot{x}_{i} =  - Bx_{i}^{f} + \sum_{j=1}^{N} A_{ij} \frac{x_{j}^{h}}{x_{j}^{h}+1}$\\
Epidemic & Probability of Infection& $\dot{x}_i=-x_{i} + \sum_{j=1}^{N}A_{ij}B (1-x_i)x_j$\\
\end{tabular}
\end{ruledtabular}
\end{table*}

\subsection{Mean-field Approach}

We define a steady state approximation $ {x}_{i}^{\text{mfa}}(\param)\approx \steady_i (\param)$ based on a decoupled version of Eq.~(\ref{eq:ode}) and seek the minimizer of the following surrogate objective function
\begin{align}
\obj^{\text{mfa}}(\boldsymbol{\theta},\mathbf{y}) =\frac1N\sum_{i=1}^N(\steady_i^{\text{mfa}}(\param)-y_i)^2. \label{eq:obj_mfa}
\end{align}
The derivation of $\steady_i^{\text{mfa}}$ uses a reduction of the $N$-dimensional ODEs  to a 1-dimensional mean-field equation described in~\cite{gao2016universal}. We summarize the reduction process as follows. The derivation holds for any feasible parameter and we temporarily drop the dependence on \math{\param}.

In a network with a low degree correlation, the neighborhoods of different nodes exhibit similarities. We assume that the impact received from these neighborhoods is identical for all nodes and consider the shared neighbor effect as a sum of interaction terms weighted by the relative out-degree  of each node. This can be represented as:
\begin{align} 
  \sum_{j=1}^{N} \frac{\delta_j^{\text{out}}}{\sum_{k}\delta_k^{\text{out}}} g(x_i, x_j) 
\end{align}
The definition of the global state is reasonable because, in connected graphs, each node's information propagates throughout the entire network, influencing every node. Furthermore, nodes with larger out-degrees have a higher impact on other nodes.

To simplify notation, we introduce an averaging linear operator $\mathcal{L}: \mathbb{R}^{N}\rightarrow \mathbb{R}$, which computes a degree-weighted average:
\begin{align}
\mathcal{L}(\mathbf{x})=\frac{\mathbf{1}^{\top}A\mathbf{x}}{\mathbf{1}^{\top}A\mathbf{1}}
\end{align}
By replacing the local averaging of neighbor effects with the global impact, we obtain:
\begin{align}
\dot{x}_i &= f(x_i) + \delta_i^{\text{in}} \left[ \sum_{j=1}^{N}\frac{A_{ij}}{\delta_i^{\text{in}}}g(x_i, x_j)\right]\\
&\approx f(x_i) + \delta_i^{\text{in}} \left[\sum_{j=1}^{N} \frac{\delta_j^{\text{out}}}{\sum_{k}\delta_k^{\text{out}}} g(x_i, x_j) \right] \\
&= f(x_i) + \delta_i^{\text{in}} \mathcal{L}(g(x_i, \mathbf{x}))
\end{align}
where $g(x_i,\mathbf{x}) = [g(x_i,x_1),\ldots,g(x_i,x_N)]$.
We further approximate the term $\mathcal{L}(g(x_i, \mathbf{x}))$ by $g(x_i, \mathcal{L}(\mathbf{x}))$:
\begin{align}
\dot{x}_i&\approx f(x_i) + \delta_i^{\text{in}} g(x_i, \mathcal{L}(\mathbf{x})) \label{eq:g(x,L(x))}.
\end{align}
This approximation is exact when $g(x_i,x_j)$ is linear in $x_j$. When $g(x_i, x_j)$ is concave (convex) in $x_j$, the approximation is an over-estimation (under-estimation).
In vector form, this system is
\begin{align}
\dot{\mathbf{x}}\approx f(\mathbf{x}) + \degree^{\text{in}} \circ g(\mathbf{x},\mathcal{L}(\mathbf{x})),
\end{align}
where \math{f(\mathbf{x})=[f(x_1),\ldots,f(x_N)]}, \math{ g(\mathbf{x},\cdot)=[g(x_1,\cdot),\ldots,g(x_N, \cdot)]}, and \math{\circ} denotes entrywise product between vectors.
Applying $\mathcal{L}$ on both sides reduces this system to one dimension,
\begin{align}
\mathcal{L}(\dot{\mathbf{x}}) &\approx \mathcal{L}[f(\mathbf{x})] +  \mathcal{L}[\degree^{\text{in}}\circ g(\mathbf{x},\mathcal{L}(\mathbf{x}))].
\end{align}
Since \math{\mathcal{L}} is linear,
\math{\mathcal{L}(\dot{\mathbf{x}})=\dot{\mathcal{L}(\mathbf{x})}}.
The following approximation assumes $f,g$ are linear and $\mathcal{L}$ roughly preserves multiplication.
\begin{align}
\dot{\mathcal{L}(\mathbf{x})}&\approx f(\mathcal{L}(\mathbf{x})) + \mathcal{L}( \degree^{\text{in}}) \mathcal{L}[g(\mathbf{x},\mathcal{L}(\mathbf{x}))] \\
&\approx f(\mathcal{L}(\mathbf{x})) + \mathcal{L}( \degree^{\text{in}}) g(\mathcal{L}(\mathbf{x}),\mathcal{L}(\mathbf{x}))\label{eq:before_1dmfa}.
\end{align}
Define a global effective state $x_{\text{eff}}=\mathcal{L}(\mathbf{x})$ and an effective scalar representation of topology $\beta = \mathcal{L}( \degree^{\text{in}})$. Then, from Eq.~(\ref{eq:before_1dmfa}), $x_{\text{eff}}$
approximately follows the dynamics
\begin{align}
\dot{x_{\text{eff}}} = f(x_{\text{eff}}) + \beta g(x_{\text{eff}},x_{\text{eff}}). \label{eq:1dmfa}
\end{align}
Using \math{x_{\text{eff}}} for $\mathcal{L}(\mathbf{x})$ in Eq.~(\ref{eq:g(x,L(x))}), we get the uncoupled ODE for \math{x_i}:
\begin{align}
\dot{x}_i &= f(x_i) + \delta_i^{\text{in}} g(x_i, x_{\text{eff}} ). \label{eq:mfa}
\end{align}

The derivation is based on the following assumptions:
\begin{inparaenum}[(i)]
\item A shared global impact can approximate the local neighbor effect of each node.
\item The functions $f$ and $g$ are approximately linear, which holds when $x_i$ and $x_j$ are close to their steady states.
\item The operator $\mathcal{L}$ roughly satisfies the approximation $\mathcal{L}(\mathbf{x}\circ\mathbf{y})\approx \mathcal{L}(\mathbf{x})\mathcal{L}(\mathbf{y})$.
\end{inparaenum}
Integrating the decoupled ODEs (Eq.~(\ref{eq:mfa})) is computationally cheap compared to solving coupled systems. The mean-field steady states are accurate for a variety of network dynamics and topologies, see~\cite{gao2016universal} for a discussion. 
Essentially, we replace distinct neighbor states with a global effective state shared by all nodes. The decoupled ODE decomposes the incoming strength of node $i$ into its in-degree and a constant neighbor interaction $g(x_i, x_{\text{eff}})$. This equation captures the highly correlated nature of degrees and states in the dynamics~\cite{gao2016universal,sanhedrai2022reviving}.

There are three steps to compute the steady states $\Steady(\param)$ at an arbitrary parameter $\param$. 
\begin{inparaenum}[(i)]
\item Integrate the 1-dimensional ODE~(Eq.~(\ref{eq:1dmfa})) with initial condition $\mathcal{L}(\mathbf{y})$ to get $\steady_{\text{eff}}$. This biases $\steady_{\text{eff}}$ toward the 
``right'' zero of \math{f(x,\param)+\beta g(x,x,\param)}.
\item Solve the uncoupled ODEs in Eq.~(\ref{eq:mfa}) with initial condition $\steady_{\text{eff}}$. Let $z_i$ denote the ultimate state value. 
\item Perform $k$ numerical integration steps on the   coupled 
ODEs~(Eq.~(\ref{eq:ode})) with initial condition $z_i$ to fine-tune the
mean-field steady states. 
\end{inparaenum}
When no fine-tuning steps are performed in step (iii) ($k=0$), we get our \emph{mfa} method in Algorithm~\ref{alg:mfa}, which is not only extremely efficient but only requires node-degrees and not full topology to estimate the parameters, since the scalar topology $\beta$ can be written as:
\begin{align}
\beta = \frac{\mathbf{1}^{\top}A\degree^{\text{in}}}{\mathbf{1}^{\top}A \mathbf{1}} = \frac{\degree^{\text{out}}\cdot \degree^{\text{in}}}{\degree^{\text{in}}\cdot \mathbf{1}}
\end{align}
Note that in \emph{mfa}, one only needs to solve  a decoupled ODE  for each distinct node-degree, which reduces the system dimension to $M$ ($M$ equals the number of unique degrees). An enhanced method \emph{mfa+} ($k>0$ in step (iii)) starts from the output of \emph{mfa} and simulates the exact ODEs to get more accurate steady states (Algorithm~\ref{alg:mfap}). 
The algorithms can be adapted to networks containing disconnected nodes. In this context, the zeros of the function $f(x,\param)$ represent the stable states of these isolated nodes.   

\begin{algorithm}[h]
\begin{algorithmic}[1]
  \Require    In-degree $\degree^{\text{in}}$, Out-degree  $\degree^{\text{out}}$, Observed Global Effective State $y_{\text{eff}}$, Parameter $\param$  
  \Ensure   Steady states $\Steady^{\text{mfa}} $     
   \State $\beta \leftarrow   {\degree^{\text{out}}\cdot \degree}/{\degree^{\text{in}}\cdot \mathbf{1}}$  
   \State $\steady_{\text{eff}}\leftarrow \textsc{SteadyState}\left( x\mapsto f(x,\param)+\beta g(x,x,\param),  y_{\text{eff}}\right)$   
   \State Initialize an empty array $\degree' \leftarrow [\ ]$ and a dictionary $\textit{id}\leftarrow \{\}$\\
    $j\leftarrow 1$ 
  \For{$i=1,\ldots,\text{length}(\degree^{\text{in}})$}
    \If \textbf{not} ($\delta_i^{\text{in}}$ in $\textit{unique-degree}$)  
       \State Append $\delta_i^{\text{in}}$ to the end of \textit{unique-degree}  
       \State  $\textit{id}[\delta_i^{\text{in}}] \leftarrow j$ 
       \State  $j \leftarrow j + 1$
    \EndIf  
  \EndFor
  \State $\mathbf{z}\leftarrow \textsc{SteadyState}\left( \mathbf{x}\mapsto f(\mathbf{x} )+\degree' g(\mathbf{x}, \steady_{\text{eff}}),  \steady_{\text{eff}}\mathbf{1} \right)$  
  \For{$i=1,\ldots,\text{length}(\degree^{\text{in}})$}{ 
    $\steady_i^{\text{mfa}}\leftarrow \mathbf{z}[id[\delta_i^{\text{in}}]]$
  } 
  \EndFor
\end{algorithmic} 
\caption{\textit{mfa}}\label{alg:mfa}
\end{algorithm}

\begin{algorithm}[h]
\begin{algorithmic}[1]
  \Require $n \geq 0 \vee x \neq 0$
  \Ensure $y = x^n$
  \State $y_{\text{eff}} \leftarrow \mathbf{1}^{\top}A\mathbf{y} / \mathbf{1}^{\top}A\mathbf{1}$  
  \State $\Steady^{\text{mfa}}\leftarrow \textit{mfa}(A\mathbf{1}, \mathbf{1}^{\top}\mathbf{A},y_{\text{eff}},\boldsymbol{\theta})$   
  \State $\Steady \leftarrow \textsc{SteadyState}( \mathbf{x}\mapsto [F_{1}(\mathbf{x},\param,A) ,\ldots,F_{N}(\mathbf{x},\param,A)],  \Steady^{\text{mfa}})$ 
\caption{\textit{mfa}+}\label{alg:mfap}
\end{algorithmic}
\end{algorithm}

\begin{figure*}[t]
    \centering
    \includegraphics[width=0.8\textwidth]{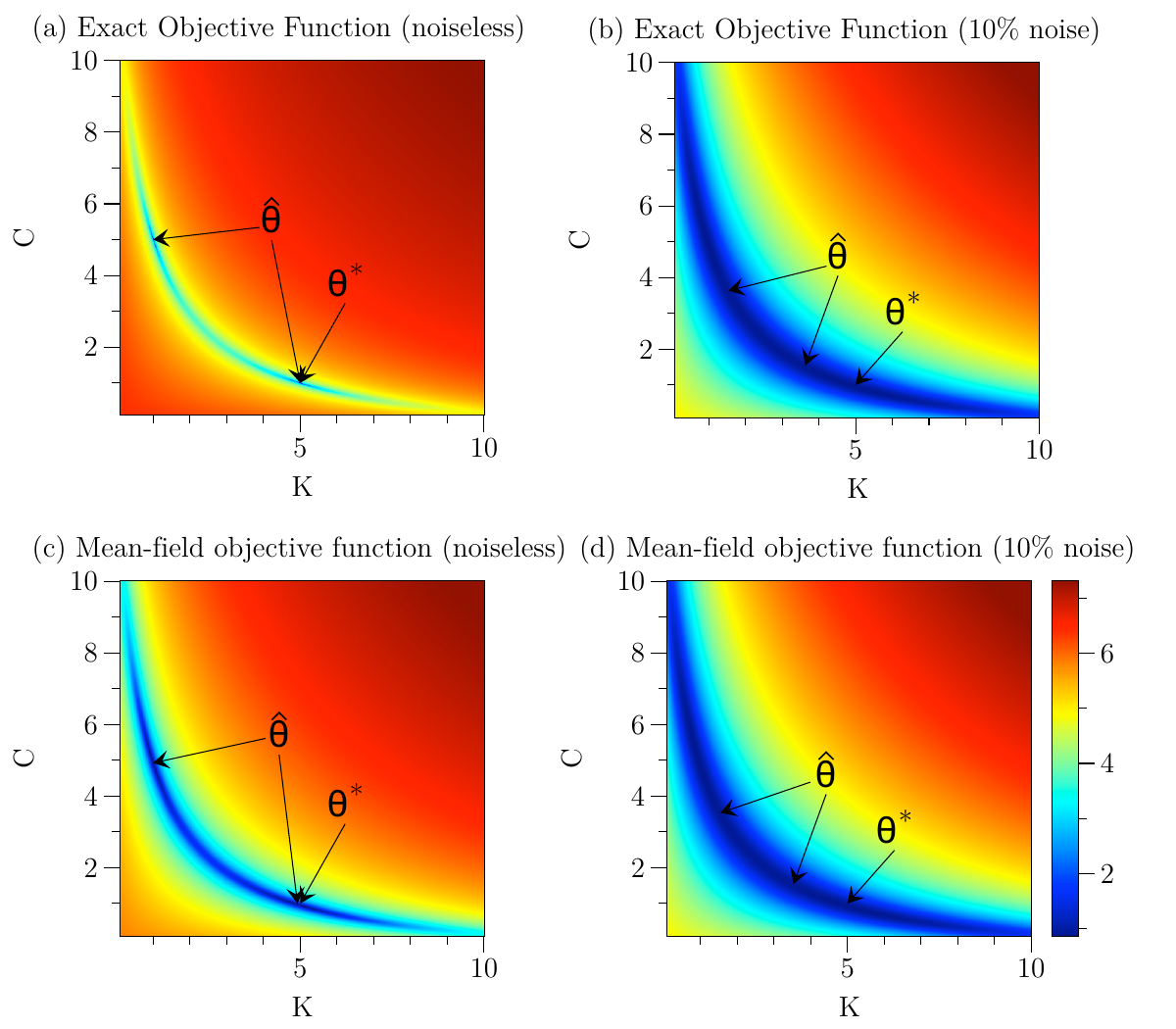}
    \caption{Loss landscapes for ecological dynamics over parameter space $\{[B^*,K,C,D^*,E^*,H^*]\ :\ K,C\in[0.01,10]\}$. The ground truth parameter used to generate $\mathbf{x}^*$ is $\param^*=[0.1,5,1,5,0.9,0.1]$. $\hat{\param}$ denotes the global minimum. The four figures correspond to different objective functions: (a) $\obj(\param,\mathbf{x}^*)$  , (b) $\obj(\param,\mathbf{y})$, (c) $\obj^{\text{mfa}}(\param,\mathbf{x}^*)$, and (d) $\obj^{\text{mfa}}(\param,\mathbf{y})$, where $y_i\sim \mathcal{N}(x_i^*, 10\%x_i^*)$. The ecological system is degenerate as $(K,C)=(5,1)$ and $(K,C)=(1,5)$ yield the same steady states, as manifested in the two basins of attractions in (a). Injecting noise to observation flattens the loss landscape around $\param^*$ and the global minimum deviates from the ground truth, making $\param^*$ ``unidentifiable''. Converging to either basin of attraction reproduces the ground truth states, which is why we use state error, not parameter error to evaluate learning. 
    }
    \label{fig:obj}
\end{figure*}
We also consider a brute-force method \emph{full} that evaluates $\Steady(\param)$ by simulating the ODE using the noisy states $\mathbf{y}$ as initial condition (Algorithm~\ref{alg:full}). \emph{Full} is a benchmark for comparing the mean-field \emph{mfa} method. Our empirical results in Section~\ref{sec:emp} show that the mean-field solver obtain similar steady states for large or near-homogeneous networks using far fewer computational resources. 
\emph{Mfa+} can be more efficient than \textit{full}, because fewer numerical integration steps are needed when starting from \emph{mfa}-steady states as opposed to starting from $\mathbf{y}$. 

The \textsc{SteadyState} function is utilized in the \textit{mfa}, \textit{mfa}+, and \textit{full} modules. It takes a user-specified ordinary differential equation (ODE), an initial condition, and a step size (optional) as inputs, providing the corresponding steady states. The ODE is iteratively evaluated until the sum of absolute derivatives of all nodes converges to zero. For the detailed implementation based on the forward Euler scheme, please see Appendix~\ref{sec:supp-ss}.

The computational complexity of these algorithms depends on the number of iterations (\#iterations) required to compute steady states and the runtime needed for a single evaluation of the differential equation. The \textit{mfa} algorithm exhibits a time complexity of $O(N \cdot \text{\#iterations})$, reflecting its efficiency in scenarios where nodes act independently.
In contrast, both the \textit{full} and \textit{mfa}+ algorithms face a more significant computational burden with a time complexity of $O(N^2\cdot\text{\#iterations})$. 

While we lack precise theoretical results on \#iterations, we seek to quantify it regarding network topology properties. In \textit{mfa}, the number of iterations required for convergence remains constant w.r.t. network size since node dynamics are entirely independent.
For the \textit{full} method, \#iterations is at least on the order of the network's diameter. To illustrate, consider an initial state where all nodes except one are in equilibrium. Achieving steady states is analogous to propagating the impact of a perturbation throughout the entire network. Each iteration executes a round of message passing to the nearest neighbors. Thus, the number of iterations required is at least the shortest path length between the two most distant nodes in the network. The concept of the number of iterations also bears a resemblance to the mixing time in Markov chains, which has a proven lower bound associated with the network's diameter~\cite{levin2017markov}.  Based on this analysis, the time complexity of the \textit{full} method on random networks is approximately $\Omega\left( {N^2\log N}/{\log \langle \degree \rangle}\right)$, while for square lattice networks, the time complexity is approximately $\Omega(N^{5/2})$. Here, $\langle \cdot \rangle$ represents the mean of the elements in the enclosed vector.

Regarding \textit{mfa}+, its running time is bounded by the numerical integration of the coupled system. Given that the initial condition is the mean-field equilibrium, which is already near the final solution, the number of iterations to execute the  coupled ODE should be less than \textit{full}. In comparison to \textit{full}, \textit{mfa}+ reaches the same equilibrium in less time by circumventing the initial propagation phase throughout the network.

Fig.\ref{fig:obj} illustrates the loss landscapes of the objective function (Eq.~(\ref{eq:obj})) and its mean-field version (Eq.~(\ref{eq:obj_mfa})). We utilize the \emph{full} method to compute $\steady_i(\param)$ and employ Algorithm \ref{alg:mfa} to calculate $\steady_i^{\text{mfa}}(\param)$. Fig.~\ref{fig:obj}(a) and (b) correspond to the exact objective function when the observation $\mathbf{y}$ is noiseless ($\mathbf{y}=\mathbf{x}^*$) and when it is blurred by 10\% Gaussian noise ($y_i\sim \mathcal{N}(x_i^*, 10\%x_i^*)$), respectively. The corresponding surrogate surfaces are shown in Fig.~\ref{fig:obj}(c-d), which approximate the shape and minimum of the true surface regardless of the noise level. This indicates that optimizing  the surrogate error would yield approximately the same solution. 

We show experimental results on the accuracy of parameters learned by minimizing $\mathcal{E}$ compared to $\mathcal{E}^{\text{mfa}}$ on various networks in section~\ref{sec:emp}. However, instead of focusing on parameter error, we assess the learning outcome based on the ability of the learned parameters to reproduce the ground truth steady states $\mathbf{x}^*$. This choice is motivated by two factors. Firstly, in the case of noiseless data and a single global minimum exists, reproducing the steady states is equivalent to recovering the ground truth parameters. However, when multiple global minima are present, as illustrated in Fig.~\ref{fig:obj}(a), the optimization process may lead to convergence in different basins of attraction, all of which equally reproduce the steady states. Secondly, when the observed data is noisy, as depicted in Fig.\ref{fig:obj}(b), the minimizer deviates from the ground truth parameter. Therefore, evaluating the learned parameters based on their ability to reproduce the noiseless ground truth steady states provides a more meaningful measure of success, as it focuses on the underlying dynamics rather than the specific parameter values.

\section{\label{sec:convergence}NLS Convergence Analysis}

We now investigate the asymptotics of NLS for denoising observed steady states given zero-mean noise. Theorem 3.2 in~\cite{magdon2001learning} proves that under zero-mean and independent noise with variance $\sigma^2$, the expected error of the learned steady states is $O(\sigma^2/N)$. If however, the variance of noise at node $i$ depends on $x_{i}^*$, which is affected by the steady states of $i$'s neighbors then the noisy states are interdependent. 

To address this interdependence, we analyze an idealized network model that breaks the dependence between nodes while approximately preserving the steady states of the original network. Specifically, we assume nodes with a fixed degree $k$ are sampled from a   $k$-regular network (undirected), where every node has degree $k$.

\begin{definition}
  \textup{Block-regular Network (BRN).} A block-regular network, denoted as $\mathcal{B}$, consists of $M$ disconnected blocks $\mathcal{B}_{1},\ldots,\mathcal{B}_{M}$. Each block $\mathcal{B}_{i}$ has  $N_{i}$ nodes sampled from a regular network with degree $r_{i}$ and $r_{i}\neq r_{j}$ for $i\neq j$. $r_{i}$ is called the block degree of $\mathcal{B}_{i}$.
\end{definition}

Theorem 2 in~\cite{jiang2020true} shows that one basin of attraction for states of a $k$-regular network is $\forall i \in \mathcal{V}, x_i^*=z$, where $z$ is a nonnegative root of 
\begin{align}
  f(x,\param^*) + k g(x,x,\param^*),\label{eq:regular}
\end{align} and $\mathcal{V}$ is the set of nodes in $\mathcal{G}$.
Hence, each node's steady state can be obtained from an independent equation.

A BRN is defined by a sequence of regular blocks of specified sizes and degrees. Analyzing nodal behavior stratified by degree has been adopted in~\cite{pastor2001epidemic} to study epidemic models. Fig.~\ref{fig:brn_vs_general} provides additional evidence, where network steady state can be approximated by the associated BRN. This means analysis of BRNs provides insight into general networks.

\begin{figure*}[t]
    \centering
    \includegraphics[width=0.9\textwidth]{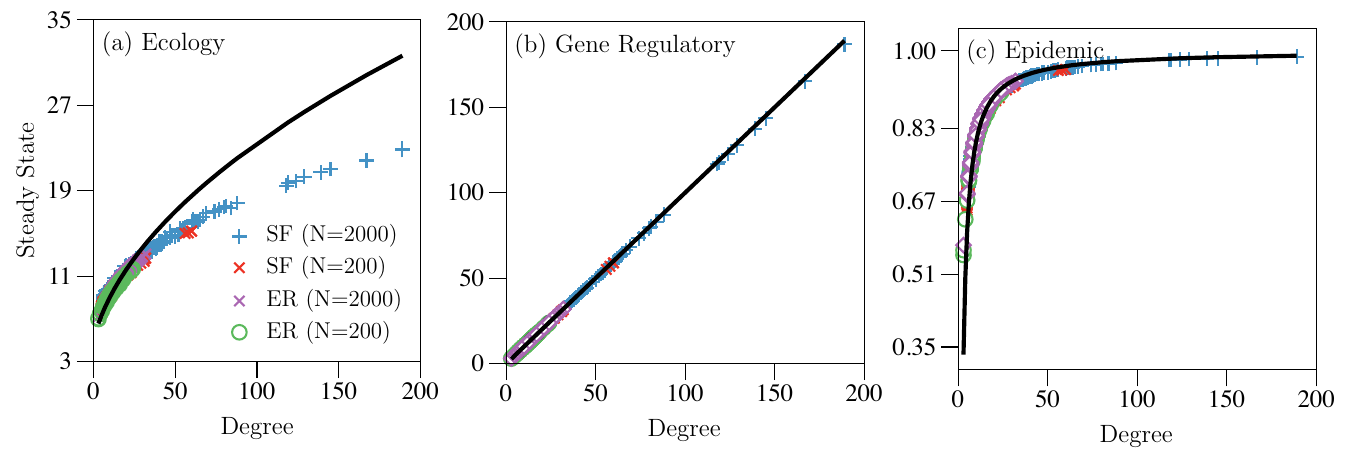}
    \caption{Ground truth versus BRN-approximated steady states. The steady states calculated by BRNs are represented by the black line, while the actual equilibria are depicted as dots.  To compute the BRN state, we assume nodes with the same degree are sampled from a regular network with degree $\delta$ and the initial condition is a high state, the steady state of degree-$\delta$ nodes converges to the non-negative root of $f(x,\param^*) + \delta g(x,x,\param^*)$ (black line) for three dynamics: (a) ecology, (b) gene regulatory, and (c) epidemic. The topology of network includes a scale-free (SF) network with $N=200$, average degree $\langle \degree\rangle=9.75$, and the power-law exponent of the degree distribution $\gamma=-2$; a scale-free network with $N=2000$, $\langle \degree\rangle=11.96$, and $\gamma=-2$; an Erd\H{o}s-R\'enyi   network with $N=200$, $\langle \degree\rangle=11.58$, and an ER network with $N=2000$, $\langle \degree\rangle=15.98$.      }
    \label{fig:brn_vs_general}
\end{figure*}

In the context of BRNs, the objective function for reverse engineering can be written as
\begin{align} 
\obj (\boldsymbol{\theta},\mathbf{y}) = \frac{1}{M} \sum_{i=1}^{M} \frac{1}{N_i}\sum_{\delta_j=r_i}  (\steady_j(\param) - y_j)^2. \label{eq:obj1}
\end{align}
The disconnected blocks are separate graphs, each having its own global effective state, which is the steady state of the 1-dimensional ODE. Consequently, the mean-field equation is equivalent to the 1-dimensional ODE for a regular block:
\begin{align}
  \dot{x}_i = f(x_i) + \delta_i g(x_i, x_{\text{eff}})  = f(x_i) + \delta_i g(x_i, x_i)
\end{align}
Therefore, the conclusion for NLS on BRNs applies to the mean-field objective.

NLS for BRNs asymptotically converges to the ground truth parameter provided that (i) the map from parameter
\math{\param} to steady states is continuous and (ii) the minimizer of the objective is unique for sufficiently many blocks or distinct degrees. By the Implicit Function Theorem, assumption (i) holds when the dynamics satisfies an invertibility condition.
\begin{lemma}
\label{lem:smooth_apx}
  Suppose $\mathcal{G}$ is a BRN, consisting of $r_{i}$-regular graphs, for $i=1,\ldots,M$.
  If 
  \begin{align}\forall i,\ \frac{\partial f}{\partial x_{i}}+ r_{i}\frac{\partial g}{\partial x_{i}}\rvert_{x_i=x_i^*}\neq 0,\label{lemma_condition}
  \end{align} 
  there exists a compact set $\Theta$ containing $\boldsymbol{\theta}^{*}$ 
  s.t. there is an invertible function mapping from parameters to steady states $\Steady(\param)$ with $\Steady(\param^*)=x_{i}^{*}$, i.e.,
  \begin{align} f(x_{i}^{*},\param^*)+
  r_{i} g(x_{i}^{*},x_{i}^{*},\param^*)=0.\end{align}
\end{lemma} 
We assume that Eq.~(\ref{lemma_condition}) holds, thereby establishing condition (i). To demonstrate condition (ii), we provide numerical evidence in the appendix. 
The following corollary follows from the asymptotic properties of nonlinear least squares~\cite{jennrich1969asymptotic}.
\begin{corollary}
\textup{(Consistency).} Assume the block degrees $r_{1},\ldots,r_{M}$ is a bounded sequence converging completely to some distribution function $\mathrm{F}(\delta)$. Then
$\hat{\boldsymbol{\theta}} \rightarrow  \boldsymbol{\theta}^{*}$ as $M\rightarrow \infty$.
\end{corollary} 

The asymptotic properties of least-squares estimators allow us to quantify the denoising property on BRN. In particular, Theorem 3.2 in~\cite{magdon2001learning} proves that the expected mean squared error is bounded by the average variance of the observed data points divided by the total number of data points.  
\begin{corollary}\label{prop:denoise}
For a Block-regular Network (BRN) consisting of $M$ blocks with sizes $N_i$, where each node has i.i.d. zero-mean noise with variance $\sigma_j^2$, the expected state error $\mathbb{E}[\obj(\hat{\boldsymbol{\theta}})]$ is on the order of $O\left(\frac{\overline{\sigma^2}}{N}\right)$, where $\overline{\sigma^2}$ is defined as the average noise variance across all nodes in the network and $N$ represents the total number of nodes in the network, given by $N = \sum_{i=1}^{M} N_i$. 
\end{corollary}  
 
For general networks, this analysis on BRNs suggests that as the number of nodes in the network increases, the learned steady states from the NLS-inferred parameters asymptotically converge to the true steady states, provided there is enough degree heterogeneity in the network. Indeed, this quantitative denoising behavior is confirmed by our empirical studies on general (non-BRN) networks. For the BRN analysis to carry over to general networks, one must still have continuity of steady states w.r.t. parameters and identifiability of \math{\param^*}. One can argue (deferred to Appendix~\ref{sup:brn}) that if these properties hold for a BRN that is sufficiently ``close'' to the real network, then the properties will continue to hold for the real network.

\section{\label{sec:emp}Numerical Results}
We assess the effectiveness of our method by measuring the deviation between the predicted steady states obtained from the learned parameters and the true steady states. 
To make a standardized comparison across different networks with varying steady-state ranges, we use the mean absolute percentage error (MAPE)
\begin{align}
\frac{1}{N}\sum_{i=1}^{N} \left\lvert\frac{\steady_i(\hat{\param}) - x_i^*}{x_i^*}\right\rvert.
\end{align}
As discussed in Section~\ref{sec:surrogateODE}, the minimizer of the noisy objective function is not necessarily the true parameter. Hence, we report the extent to which the learned parameters recover the true steady states.

We use the same ground truth parameters and initial conditions as in previous studies~\cite{gao2016universal,jiang2020true}. For ecological networks, the true parameters are set to $B=0.1, K=5,C=1,D=5,E=0.9,H=0.1$. The true parameters for gene regulatory networks are $B=f=1,h=2$, and for epidemic networks, $B=0.5$. The initial conditions of the states are $\forall\ i\in \mathcal{V},\ x_i = 6$ for ecological and gene regulatory dynamics and $\forall\ i\in \mathcal{V},\ x_i =0.5$ for epidemic networks.

Table~\ref{tab:data} lists the network data used in our experiments. The networks include two mutualistic networks (Net8, Net6)~\cite{gao2016universal}, two transcription networks of Saccharomyces cerevisiae (TYA, MEC)~\cite{gao2016universal}, a human contact network (Dublin)~\cite{rossi2015network}, and an email communication network (Email)~\cite{guimera2003self,kunegis2013konect}. These networks are unweighted and undirected. We analyze two types of synthetic networks: Erd\H{o}s-R\'enyi (ER) and scale-free (SF). Their degree distribution follows Poisson and Power law distribution, respectively.

\begin{table}[ht]
\caption{\label{tab:data} The real network data analyzed in the paper. For each network, we show its number of nodes and edges, dynamics, and the mean ($\langle \mathbf{x}^{*}\rangle$) and standard deviation ($\sigma$) of the ground truth steady states.  } 
\begin{ruledtabular}
\begin{tabular}{lllll} 
Net & Dynamics & \# nodes & \# edges  &  $\langle \mathbf{x}^{*}\rangle \pm \sigma$ \\\colrule
Net8   & Ecology    & 97   & 972  & 11.4 $\pm$ 2.7 \\
Net6   & Ecology    & 270  & 8074 &  9.5 $\pm$ 0.74\\
TYA    & Regulatory & 662  & 1062 &2.8 $\pm $ 4.3\\
MEC    & Regulatory & 2268 & 5620 &  11.48 $\pm$ 3.36 \\
Dublin & Epidemic   & 410  & 2765 & 0.78 $\pm$ 0.14\\
Email  & Epidemic   & 1133 & 5451 &  0.82 $\pm$ 0.05   
\end{tabular} 
\end{ruledtabular}
\end{table}

We adopt Conjugate Gradient Descent to perform parameter searching. To compute steady states, we run the Runge-Kutta4 scheme implemented by running  ODEINT~\cite{virtanen2020scipy} of Scipy until the $L_1$ norm of the absolute derivative ($\lvert d \Steady(\param)/dt\rvert_1$) is close to zero.  
Our implementation for all experiments and the network data can be accessed online~\footnote{\url{https://github.com/dingyanna/param_estimation.git}}.

\subsection{Efficiency} 
We compare the performance of two variants of our method and the {\em full} algorithm using real networks in Table~\ref{tab:efficiency} and Fig.~\ref{fig:runtime}. We generate observations from a normal distribution with a mean $x_i^*$ and a standard deviation $0.1\lvert x_i^*\rvert$, with the initial parameter guess   randomly sampled from a uniform distribution. Our {\em mfa} method demonstrated superior efficiency to the other approaches while maintaining a slight loss of accuracy for ecological and epidemic networks. Furthermore, our enhanced variant, {\em mfa}+, achieved comparable accuracy to the {\em full} algorithm in a slightly shorter time frame.  
When applied to ecological networks, optimizing the {\em full} objective with the initial parameter guess obtained from {\em mfa} achieves the same performance as directly optimizing the {\em full} objective, but in less than one-tenth of the time.

The accuracy of the mean-field objective function approximation relies on the dynamics and topology of the system. The surrogate objective function effectively captures the exact objective for networks characterized by linear dynamics and homogeneous degree distributions. Consequently, the relative steady-state error is lower for ecological and epidemic networks than for gene regulatory networks.

\begin{table}[h]
\centering
\caption{This table presents the runtime (in seconds) and mean absolute percentage error of steady states (in \%). We evaluate three distinct parameter estimation algorithms, each employing different methodologies to compute steady states within the objective function. The reported values represent the mean and standard deviation of the corresponding metrics.
}
\label{tab:efficiency}
\begin{ruledtabular}
\begin{tabular}{llccc}
 \multirow{2}{*}{Net}   &  \multirow{2}{*}{Metric} & \multicolumn{3}{c}{Method}  \\
 & & {\em full} & {\em mfa+} & {\em mfa} \\
  \colrule
 \multirow{2}{*}{Net8} & Runtime &  $267\pm16$    & $221  \pm  25 $ & $29 \pm   1 $  \\   
  & MAPE & $2.51 \pm 0.21 $ & $2.39\pm 0.25$ & $2.85\pm 0.16$ \\
  \colrule
  \multirow{2}{*}{Net6}  & Runtime   & $8558 \pm 717$ & $6516 \pm 1112 $ & $107\pm 7 $  \\
  & MAPE & $1.6\pm 0.14 $ & $1.63 \pm 0.16 $ & $2.7\pm 0.11$  \\ 
  \colrule
  \multirow{2}{*}{TYA} & Runtime    & $369\pm 21 $ & $373\pm 32$ & $67\pm 6$  \\ 
  & MAPE & $3.53\pm 0.29 $ & $3.68\pm 0.35$ & $8.53 \pm 0.16 $\\ 
   \colrule
  \multirow{2}{*}{MEC} & Runtime   & $2456\pm 188$ & $1161\pm 196$ & $97\pm 8$  \\    
  & MAPE & $4.76\pm 0.44$& $4.42\pm0.51$ & $6.45\pm 0.2$\\
  \colrule
   \multirow{2}{*}{Dublin}& Runtime  & $11.77 \pm 0.82$ & $10.49  \pm 0.92 $ & $1.09 \pm 0.03 $  \\  
   & MAPE & $0.47\pm0.061$& $0.47\pm0.061$&$2.42\pm 0.11$  \\  
  \colrule
 \multirow{2}{*}{Email} & Runtime   & $56.08  \pm 4.33 $& $68.52  \pm 4.01 $& $1.43 \pm 0.05 $  \\  
 & MAPE  & $0.33\pm 0.04$ & $0.33\pm 0.04 $ & $4.01\pm 0.08 $ \\
\end{tabular}
\end{ruledtabular}
\end{table} 
\begin{figure}[h]
  \centering
  \includegraphics[width=0.6\textwidth]{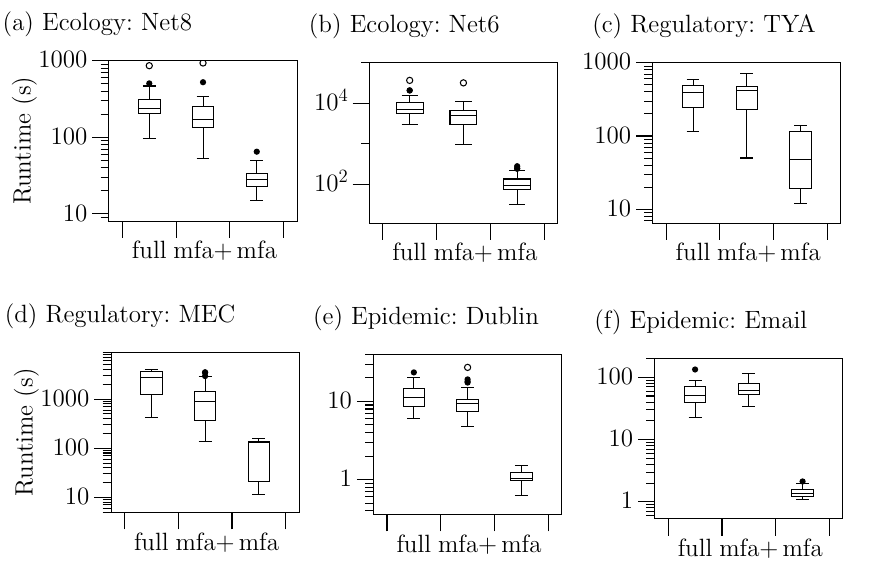}
  \caption{
    The runtime performance of three distinct parameter estimation methods was evaluated across six real-world networks. To visualize these results, we employed box-and-whisker plots using the Tukey method. The box represents the interquartile range (IQR), and a line inside the box marks the median. The whiskers extend 1.5 times the IQR from the edges of the box, and any outliers are also displayed. 
   }
  \label{fig:runtime}
\end{figure}

To assess the computational efficiency, we compare the \textit{mfa} and \textit{full} methods by simulating  steady states for ER networks of various sizes with a fixed average degree, using initial conditions and parameters defined earlier. The speedup, computed as the runtime of \textit{full} divided by that of \textit{mfa}, approximately scales at $O(N)$ for all three dynamics, with a greater scaling constant observed for ecological dynamics (Fig.~\ref{fig:speedup}). 
This finding highlights the notable computational efficiency of \textit{mfa}, particularly for large-scale networks, as it decouples the ODEs and accelerates the parameter estimation process.

\begin{figure}[h]
  \centering
  \includegraphics[width=0.4\textwidth]{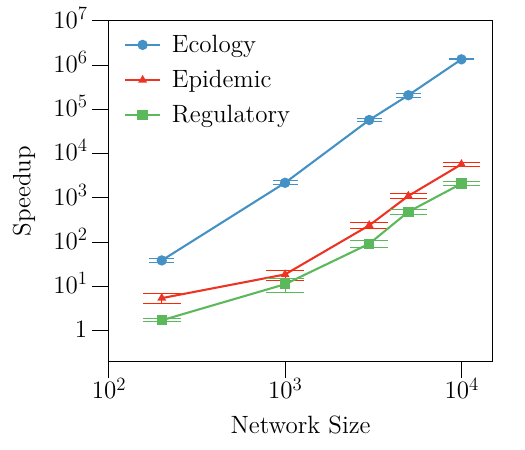}
  \caption{
    The speedup of \textit{mfa} compared to \textit{full} in computing a single steady state of Erd\H{o}s-R\'enyi networks with various sizes. The networks have a fixed average degree 12 and vary in size, with the sizes of 200, 1000, 3000, 5000, and 10000. Speedup is computed as the ratio of the runtime of the \textit{full} method to that of the \textit{mfa} method. Each data point represents the mean result, with error bars depicting the standard deviation from 43 runs (for regulatory and epidemic) and 5 runs (for ecology).
   }
  \label{fig:speedup}
\end{figure}

\begin{figure*}[t]
  \centering
  \includegraphics[width=\textwidth]{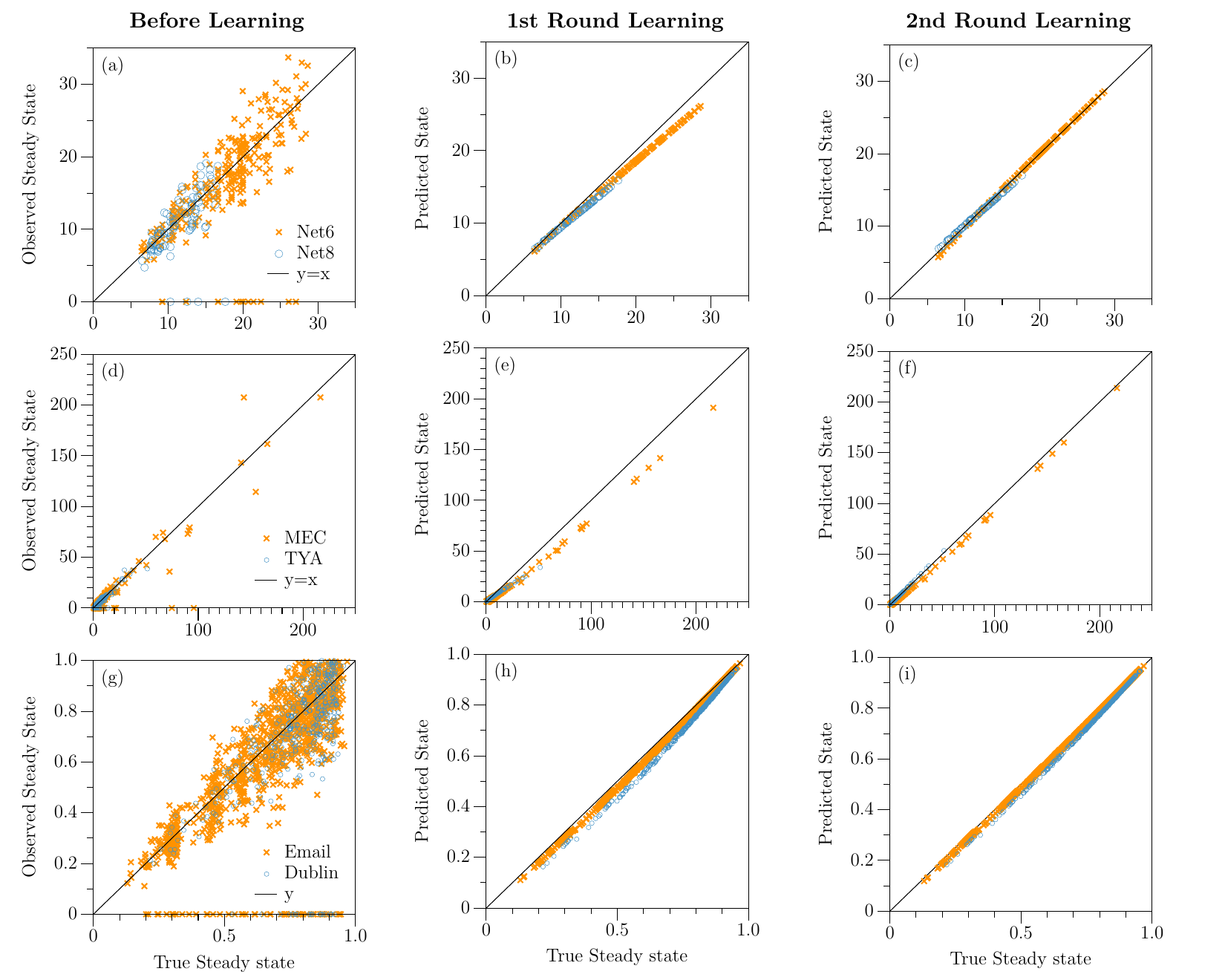}
  \caption{This figure shows the estimated states versus true steady states for six real networks (ref. Table~\ref{tab:data}) governed by ecological (a-c), gene regulatory (d-f), and epidemic (g-i) dynamics, where parameters are learned given 5\% states mismeasured to zero and other states with Gaussian noise of level 13\%. We plot the estimated states before learning (a,d,g), after the first run of optimization using the surrogate surface $\mathcal{E}^{\text{mfa}}$ (b,e,h), and after the second run of learning via the {\em mfa}+ method to compute steady states (c,f,i). The first run of learning identifies the nodes that are completely mismeasured. Then, a new objective function summing over non-mismeasured states is optimized to locate a new parameter, generating the steady states after the second round of learning. 
   }
  \label{fig:adversarial}
\end{figure*}

\subsection{Denoising Observed Steady States} 
We assess the performance of our methods under various types and levels of data distortion. Specifically, we consider Gaussian noise, where the observations are generated according to $y_i \sim \mathcal{N}(x_i^*, \epsilon \lvert x_i^* \rvert)$ ($\epsilon\in[0,0.3]$).
We evaluate the effectiveness of the \textit{mfa} method on ER and SF networks. For ER networks, we observe that the predicted state error remains smaller than the observation error even as the noise level increases from 0 to 30\%. However, for scale-free networks, the steady-state error is reduced compared to the observation error only when the noise level exceeds 3\% (see Fig.~\ref{fig:robustness}). In cases where the noise level is minimal, and the network is heterogeneous, the parameters learned from \textit{mfa} may require fine-tuning using an enhanced version such as \textit{mfa}+. It is worth noting that despite the presence of Gaussian noise, the approximation of the minima of the surrogate surface to the exact error surface remains largely unaffected, indicating a robust alignment between the two surfaces.  
Furthermore, our application of the mean-field approach on directed ER networks with uniformly distributed link weights $\mathcal{U}[5,6]$ yields successful network state recovery, as demonstrated in Fig.~\ref{fig:directed-weighted}. 
\begin{figure}[h]
    \centering
    \includegraphics[width=0.32\textwidth]{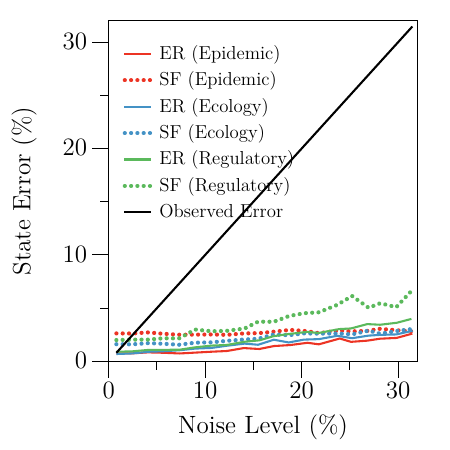}
    \caption{The relative steady-state prediction error is evaluated for synthetic networks under Gaussian noise. The black line represents the observation error caused by the noise. The solid lines correspond to an Erd\H{o}s-R\'enyi network with 200 nodes and 811 edges, while the dotted lines correspond to a scale-free network with 200 nodes and 784 edges. We minimize the mean-field objective function using an initial parameter guess sampled from $\mathcal{N}(\param^*, 20\% \lvert \param^*\rvert)$. Despite the increasing error in the observed equilibrium, the learned parameter yields states that closely approximate the ground truth steady states. The results are averaged over 100 runs.  }
    \label{fig:robustness}
\end{figure}

\begin{figure}[h]
    \centering
    \includegraphics[width=0.6\textwidth]{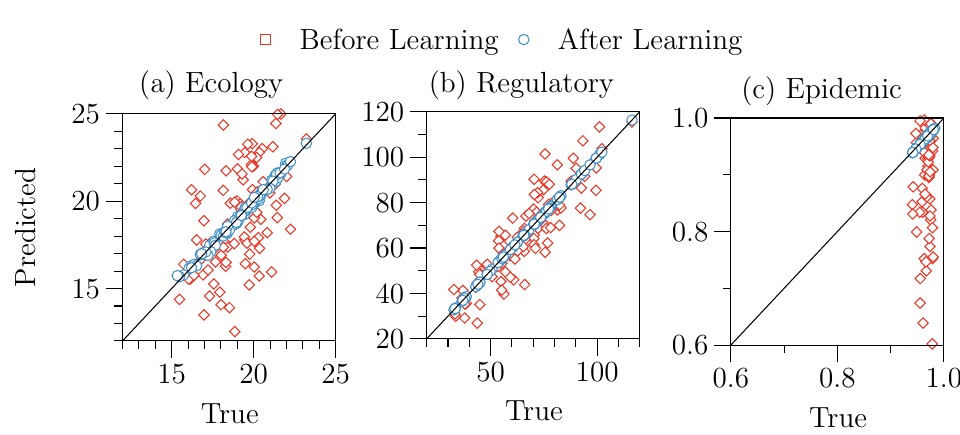}
    \caption{Denoising of steady states in directed and weighted ER networks with 100 nodes, an average degree of 12, and an average link weight of 5.5. The comparison between observed and true steady states is illustrated by the red square, while the inferred versus true equilibrium is depicted by the blue dots, with the former being simulated using the learned parameter.\label{fig:directed-weighted}}
     
\end{figure}

We simulate mismeasuring of data by setting 5\% of nodes to zero and adding Gaussian-distributed noise $\mathcal{N}(0,13\% x_{i}^*)$ to other nodes. The NLS with \textit{mfa} restores the states of the contaminated nodes (Fig.~\ref{fig:adversarial}). 
The vanishing states are recovered, and the relative state error is reduced from 16\% to 5\%. The result guides a second round of learning by signifying the mismeasured nodes whose measurement differs drastically from prediction. We identify such nodes as the set $\mathcal{V}'=\{i\in \mathcal{V} :\ \lvert (y_{i} - \steady_{i}) / (y_i+10^{-8}) \rvert \geq 1\} $. We discard the nodes in $\mathcal{V}'$ and optimize the following objective 
\begin{align}
\sum_{i \not\in \mathcal{V}'} (\steady_{i}(\param) - y_{i} )^2.
\end{align}
The second round of learning lowers the state error to 2\%.

\begin{table*}[t]
\caption{\label{tab:perturb} Relative steady-state error (\%) for perturbed networks $\mathcal{G}'$ obtained by rewiring $k\%$ edges, removing $k\%$ nodes or links uniformly at random from the original network $\mathcal{G}$. We consider four levels of perturbation: 10\%, 20\%, 30\%,  and 40\%.    The steady states of $\mathcal{G}'$ are computed using the parameter learned via optimizing $\mathcal{E}^{\text{mfa}}$ for the original network $\mathcal{G}$. The learned parameter can predict steady states after topology changes. When computing the average relative error for gene networks, we removed nodes whose ground truth states are smaller than 0.001. The average absolute error $\langle \lvert\hat{\mathbf{x}}-\mathbf{x}^*\rvert \rangle$  of these nodes with small steady states is 0.0008.
}
\begin{ruledtabular}
\begin{tabular}{lllll|llll|llll}
 \multirow{2}{*}{Net}& \multicolumn{4}{c}{Rewiring}& \multicolumn{4}{c}{Node Removal} & \multicolumn{4}{c}{Link Removal}\\
 \cmidrule(r){2-5}  \cmidrule(r){6-9} \cmidrule(r){10-13} 
 &10\% &20\% &30\% &40\%&10\% &20\% &30\% &40\%&10\% &20\% &30\% &40\%  \\
\hline
Net8   & 1.02  & 0.95 & 0.89 & 0.87 & 1.84 & 2.53 & 3.34 & 3.63 & 1.55 & 1.96 & 2.64 & 3.32\\
Net6   & 0.82 & 0.56 & 0.48 & 0.66 & 1.33 & 1.28 & 2.27 & 3.2 & 1.26 & 1.33 & 2.06 & 3.04\\
TYA    & 3.49 & 4.04 & 4.15 & 4.28 & 3.77 & 5.16 & 6.5 & 6.99 & 3.74 & 4.52 & 5.2 & 5.91  \\
MEC    & 3.09 & 2.97 & 2.85 & 2.6 & 3.53 & 7.37 & 9.94 & 15 & 3.8 & 4.6 & 6.06 & 8.16 \\
Dublin & 4.23 & 3.99 & 3.78 & 3.61 & 5.2 & 5.85 & 6.5 & 7.43 & 5.12 & 5.78 & 6.47 & 7.45 \\
Email  & 5.64 & 5.36 & 5.1 & 4.9 & 6.3 & 6.8 & 7.22& 8.4 & 6.4 & 6.77 & 7.23 & 7.95
\end{tabular}
\end{ruledtabular}
\end{table*}

\subsection{Predicting after Network Changes}
Genes interact with each other in complex ways that can vary based on changes in their physical connections. These changes in topology can cause variations in gene expression levels. Similarly, social connections that enable the spread of viruses are not fixed and can change over time. In ecological networks, introducing new species or the extinction of existing ones can also change the network's behavior. One of the primary objectives of studying network dynamics is to achieve control over the system and predict its future equilibrium states. Our work has contributed to this field by demonstrating the ability of learned parameters to predict new steady states in the face of network changes, thereby advancing the goal of controlling networks to achieve desired outcomes.

We consider three types of perturbations. (i) Random rewiring: 
we randomly delete $k\%$ of existing edges and connect $k\%$ of non-existing edges.  (ii) Node removal: we randomly remove $k\%$ of nodes and any incident edge. (iii) Link removal: we remove $k\%$ of edges uniformly at random. We use the parameter $\hat{\param}$ learned from noisy observation of level 10\% with method $\textit{mfa}$ to predict the steady states after topology changes. Using the same ground truth parameters, we simulate the benchmark steady states of the original network $\mathcal{G}$ and the perturbed one $\mathcal{G}'$. The mean relative error in approximating the steady states of $\mathcal{G}'$  is used to evaluate the prediction. 

The learned parameters successfully predict true steady states for $\mathcal{G}'$ (Fig.~\ref{fig:perturb}).
The results for four perturbation levels (10\%, 20\%, 30\%, 40\%) are shown in Table~\ref{tab:perturb}. The state error remains relatively constant and small for almost all situations. The parameters learned on a system can predict other networks governed by the same dynamics. The prediction is robust against perturbation that roughly preserves the degree distribution. 

\begin{figure}[h]
  \centering
  \includegraphics[width=0.32\textwidth]{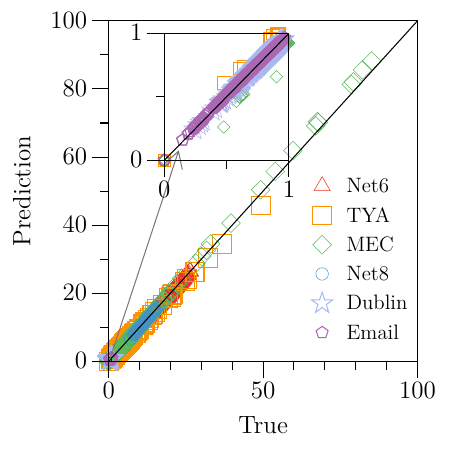}
  \caption{Predicting steady state versus true steady state after rewiring 10\% of the edges. The inferred states are computed using a parameter searching algorithm with the \textit{mfa} approach.
  The prediction  approaches ground truth steady state computed using $\param^*$ on the perturbed topology.}
  \label{fig:perturb}
\end{figure}

\subsection{Topology's Role in Learning Parameters}

We investigate the contribution of heterogeneity of degrees and number of nodes with the same degree to the denoising effect for BRNs and ER models. The performance is measured using an improvement in steady-state reconstruction:
\begin{align}
\text{Improvement} &= \frac{(\mathrm{MSE}(\mathbf{y}) - \mathrm{MSE}(\hat{\Steady}))}{\mathrm{MSE}(\hat{\Steady})}
\end{align}
where $\hat{\Steady}$ is the steady state simulated by the learned parameter $\hat{\param}$ and $\mathrm{MSE} $ is the mean squared error. The improvement ranges from $-1$ to $\infty$. If the value is in $[-1,0]$, the estimation is worse than the observation. A greater value reflects a larger denoising effect.

By Corollary~\ref{prop:denoise}, the state error of BRNs can be reduced by increasing the number of observed nodes through either the number of blocks $M$ or the block size $N_i$. 
We fix block size $N_i=20$ and tune the number of blocks $M$ from 5 to 25 (see blue curve in Fig.~\ref{fig:topology}). The block degrees are sampled uniformly from $U[5,185]$. Alternatively, we fix the block degree sequence at $[7, 13, 17, 23, 29]$ and vary block size $M$ from 20 to 100 (see red curve in Fig.~\ref{fig:topology}). For an ER network, we raise the network size by either increasing the number of distinct degrees (analogous to block number $M$) or the average number of nodes with the same degree (analogous to block size $N_i$). As we increase the number of observed nodes, the improvement steadily increases for BRNs. The ER curves are noisy versions of the BRN curves. 

Fig.~\ref{fig:network_size_mfa} shows that learned parameters offer a greater improvement for more homogeneous networks. In particular, ER networks with higher density or SF networks with smaller power-law exponents are more homogeneous and have an overall more significant improvement.  
\begin{figure}[h]
  \centering
  \includegraphics[width=0.6\textwidth]{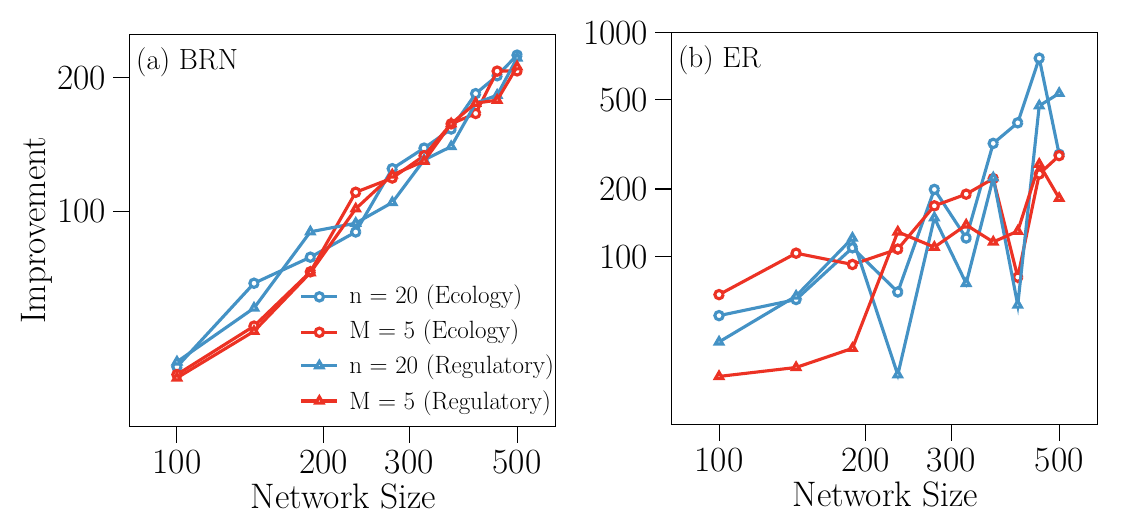}
  \caption{
    This log-log plot demonstrates the improvement in steady-state approximation as the number of observed nodes increases. We sample network sizes from the range $[100,500]$ and fix the observed noise level at 10\%. For BRNs, we obtain the parameters by minimizing the exact objective function Eq.~(\ref{eq:obj1}), while for general networks, we learn the parameters using {\em mfa}+ (\ref{alg:mfap}).  We increase network size by (i) fixing the number of nodes with the same degree at $N_i=20$ and increasing the number of distinct degrees $M$; (ii) fixing $M=5$ and varying $N_i$ from 20 to 100.  }
  \label{fig:topology}
\end{figure}

\begin{figure}[h]
    \centering
    \includegraphics[width=0.6\textwidth]{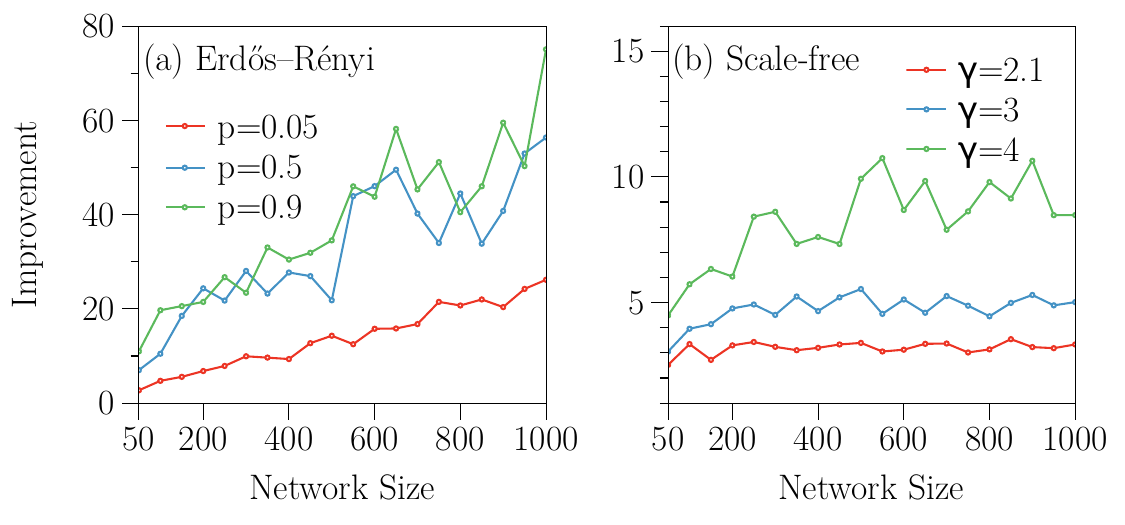}
    \caption{ We use ecological dynamics as an example to investigate the influence of network topology on the accuracy of learned parameters. We sample network sizes from the range of $[50,1000]$ and vary the density $p=\langle \degree \rangle / (N-1)$ for ER networks and the power exponent  $\gamma$ for SF networks. For each network, we learn a parameter from noisy steady states (with 10\% noise level) using NLS with \textit{mfa}.  
    }
    \label{fig:network_size_mfa}
\end{figure}

\section{Conclusion}
We presented an NLS framework to efficiently infer dynamical parameters in complex systems from steady-state data via a mean-field approach. 
The learned dynamics can \textit{denoise} erroneous observations. Our method is robust against mismeasured data, for example, steady states set to 0. The ability to recover completely mismeasured states is a denoising feature that could be of extreme interest to practitioners. For example, one can correct mistakenly measuring a species in an ecological network as tending to extinct; or, falsely reporting the sentiment of a person in a social network. The learned parameters are also able to accurately predict the new steady states if the network undergoes topology changes. This predictive capability empowers intelligent network control by adding or deleting links to guide the system toward desired equilibria. Additionally, the mean-field approach remains highly practical and applicable to large-scale networks, where computational efficiency is crucial.  

Although the mean-field surrogate objective is efficient to compute, it may not accurately capture all types of dynamics or network topologies. Particularly in situations where the dynamics exhibit high nonlinearity or the network's degree distribution is highly heterogeneous, relying solely on the mean-field objective may not fully capture the underlying dynamics. 
In these cases, the mean-field solution provides a sensible initial parameter guess that can be further refined by optimizing the exact objective function. 
 
We've theoretically established the feasibility of a simplified network model that captures modular networks formed from the union of nearly homogeneous communities. This network model extends to encompass general modular networks with restricted inter-block connections. Treating each module as an individual network allows us to capture the global state and apply the mean-field approach separately to each block. However, an increase in connections between modules or an escalation in heterogeneity within a block might lead to a potential decrease in accuracy. Increasing connections would necessitate a more intricate method to account for interactions between the modules~\cite{ma2023generalized}.

\section{ACKNOWLEDGMENTS}
 We acknowledge the support of National Science Foundation under Grant No. 2047488.

\appendix 

\section{Steady-state Computation\label{sec:supp-ss}}

The algorithm box for \textit{full} and \textsc{SteadyState} is provided in this section.
The function \textsc{SteadyState} is used in \textit{mfa}, \textit{mfa}+, and \textit{full} to compute the equilibrium given a system of ODE, an initial condition, and an optional argument of step size. The ODE is evaluated repeatedly until the sum of absolute derivatives of all nodes is close to zero.
We show an implementation for \textsc{SteadyState} based on the forward Euler scheme as follows. 

\begin{algorithm}[h]
\begin{algorithmic}[1]
  \Require {Adjacency matrix $A$, Observed steady states $\mathbf{y}$, Parameter $\param$}
  \Ensure {Steady states $  {\mathbf{x}} $ }   
   \State $ {\mathbf{x}} \leftarrow  \textsc{SteadyState}(  \mathbf{x}\mapsto [F_{1}(\mathbf{x},\param,A) ,\ldots,F_{N}(\mathbf{x},\param,A)],  \mathbf{y})$
\end{algorithmic}
\caption{\textit{full}}\label{alg:full}
\end{algorithm}  

\begin{algorithm}[h]
\begin{algorithmic}[1]
\Require{ ODE $\phi:\mathbb{R}^{n}\rightarrow \mathbb{R}^{n}$, Initial Condition $\mathbf{x}^0\in\mathbb{R}^{n}$, Step Size $h$ (default = 0.001)}
\Ensure{ Steady states $\mathbf{x} $ }  
\State $\tau \leftarrow 0$ 
\While{$ \lvert \phi(\mathbf{x}^{\tau})\rvert_1   =0$ }
    \State { 
    $\mathbf{x}^{\tau+1}\leftarrow \mathbf{x}^{\tau} + h\phi(\mathbf{x}^{\tau})$  
    \State$\tau\leftarrow \tau + 1$
}
\EndWhile
\State $\mathbf{x} \leftarrow \mathbf{x}^\tau$
\end{algorithmic}
\caption{\textsc{SteadyState}}\label{alg:ss}
\end{algorithm}
Here, $\lvert \cdot \rvert_1 $ denotes the $L_1$ norm.

\section{Toy Example for Construction of Surrogate Surface}
Consider a network with 4 nodes $\mathcal{V}=\{1,2,3,4\}$ and four edges $\{(1,2),(1,3),(1,4),(2,3)\}$. Assume the nodal activity is governed by epidemic dynamics with infection rate $B$, the states can be described by 
\begin{align}
\begin{cases} 
\dot{x}_1  = -x_1 +  Bx_2(1-x_1) +  Bx_3(1-x_1) +  Bx_4(1-x_1)  \\ 
\dot{x}_2  = - x_2 + Bx_1(1 - x_2)+  Bx_4 (1-x_2)   \\
\dot{x}_3  = -x_3 +   Bx_1(1-x_3) +B x_4 (1-x_3)  \\ 
\dot{x}_4  = - x_4 + B x_1(1-x_4)  
\end{cases}\label{eq:supp-ode}
\end{align}
Consider the ground truth parameter $B^*=1$. The steady states of this system is $$\mathbf{x}^*=[0.58750594, 0.37008106, 0.52709822, 0.52709822].$$ Assume the noisy observation is $$\mathbf{y}=[0.62091208, 0.3797219,  0.53304093, 0.51694883].$$ 
The exact objective function is written as 
\begin{align}
\mathcal{E}(B) &= \sum_{i=1}^{4}(\steady_i(B) - y_i)^2. 
\end{align}
where $ {x}_i(B)$ is the nontrivial steady state of System (\ref{eq:supp-ode}) at $B$. The states $\steady_i(B)$   can be numerically computed using Runge-Kutta4 until the mean absolute derivative is close to zero. 
For a generic parameter $B$, the nontrivial global effective state satisfies the equation 
\begin{align} 
x_{\text{eff}} = \frac{\beta_{\text{eff}}B-1}{\beta_{\text{eff}} B} 
\end{align}
where $\beta_{\text{eff}} = \frac{\sum_{i=1}^{4}\delta_i^2}{\sum_{i=1}^{4}\delta_i} = \frac{18}{8}=2.25$. 
The following formulas denote the mean-field approximation of the steady state:
\begin{align} 
\begin{cases}
\dot{x}_1  = -x_1 + 3 \frac{2.25B-1}{2.25B} (1-x_1)\\ 
\dot{x}_2  = -x_2 + 2 \frac{2.25B-1}{2.25B} (1-x_2) \\ 
\dot{x}_3  = -x_3 + 2 \frac{2.25B-1}{2.25B}  (1-x_3) \\ 
\dot{x}_4  = -x_4 +  \frac{2.25B-1}{2.25 B}  (1-x_4) \\ 
\end{cases}
\end{align}
Let $\steady_i^{\text{mfa}}(B)$ the mean-field approximation of the steady state at $B$, i.e., the root of the above equations. We can write 
\begin{align*}
&\steady_1^{\text{mfa}}(B) = \frac{6.75B - 3}{9B - 3},\ 
\steady_2^{\text{mfa}}(B) = \frac{4.5B - 2}{6.75B-2}, \\
&\steady_3^{\text{mfa}}(B) = \frac{4.5B - 2}{6.75B-2},  \ 
\steady_4^{\text{mfa}}(B) = \frac{2.25B - 1}{4.5B-1} 
\end{align*}
The surrogate objective function is 
\begin{align}
\mathcal{E}^{\text{mfa}}(B) &= \sum_{i=1}^{4}(\steady_i^{\text{mfa}}(B) - y_i)^2. 
\end{align}
In the following graph, we plot the surrogate objective function value w.r.t. parameter $B$. For comparison, we show the exact objective function in the same plot. The mean-field surrogate objective function resembles the the exact error function. More importantly, the minimizer of the surrogate objective function (1.008) is close to the minimizer of the exact objective function (1.0223). Optimizing the surrogate objective function yields a similar output to the exact. 
\begin{figure}[h]
  \centering
  \includegraphics[width=0.35\textwidth]{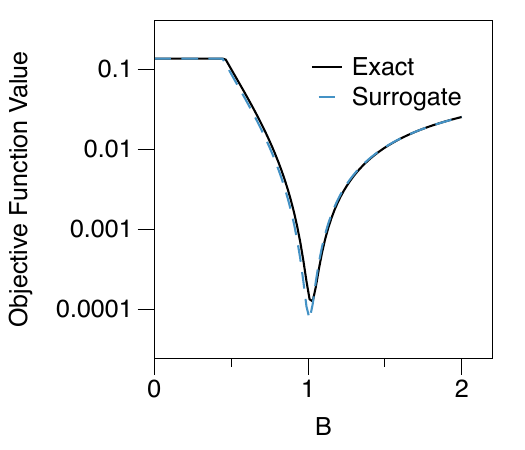}
  \caption{The comparison between the exact objective function and its surrogate version derived from the 4-node toy graph. The y-axis is presented in a logarithmic scale to emphasize the minimum.
   }
  \label{fig:surrogate}
\end{figure}

\section{Gradient Computation \label{sec:grad}}
We derive gradients for the surrogate objective function:
\begin{align}
\obj^{\text{mfa}}(\boldsymbol{\theta},\mathbf{y}) =\frac1N\sum_{i=1}^N(\steady_i^{\text{mfa}}(\param)-y_i)^2. 
\end{align} 
The gradients can be computed via   backpropagation.
\begin{align}
\frac{\partial \mathcal{E}^{\text{mfa}}}{\partial \param} = \frac{2}{N}(\Steady^{\text{mfa}}(\param) - \mathbf{y}) \frac{\partial \Steady^{\text{mfa} }}{\partial \param }
\end{align}
The mean-field steady states $\Steady^{\text{mfa}}$ satisfies
\begin{align}
 &\mathbf{F}(\Steady^{\text{mfa}},\param)\coloneqq f(\Steady^{\text{mfa}}, \param) + \degree^{\text{in}} g(\Steady^{\text{mfa}}, \hat{x}_{\text{eff}}, \param)=\mathbf{0}\\
 &\mathbf{G} (\steady_{\text{eff}},\param) \coloneqq  f(\steady_{\text{eff}},\param)+g(\steady_{\text{eff}}, \param)=0
\end{align}   
By Implicit Function Theorem,
\begin{align}
\frac{\partial \Steady^{\text{mfa} }}{\partial \param }
&=  -\left[\frac{\partial \mathbf{F}}{\partial \Steady^{\text{mfa}}}\right] ^{-1}\frac{\partial \mathbf{F}}{\partial \param } 
\end{align}
Since each nodal mean-field ODE is independent, $\frac{\partial \mathbf{F}}{\partial \Steady^{\text{mfa}}}$ is a $N\times N$ diagonal matrix  with nonzero entries ($i=1,\ldots,N$):
\begin{align}
\left[\frac{\partial \mathbf{F}}{\partial \Steady^{\text{mfa}}}\right]_{ii} = \left.\frac{\partial f}{\partial x_i  }\right\rvert_{(x_{i}^{\text{mfa}},\param)} + \delta_i^{\text{in}} \left.\frac{\partial g}{\partial x_i } \right\rvert_{(x_{i}^{\text{mfa}},\steady_{\text{eff}},\param)}
\end{align}
where the first term on the RHS is the derivative of self-dynamics $f$ w.r.t. the first argument evaluated at ($x_{i}^{\text{mfa}}, \param$) and the second term denotes the derivative of interaction function $g$ w.r.t. the first argument evaluated at $(x_{i}^{\text{mfa}},\steady_{\text{eff}},\param)$.

$\frac{\partial \mathbf{F}}{\partial \param } $ is a $N \times d$ matrix with $d$ is the number of parameters. 
\small\begin{align}
\left[\frac{\partial \mathbf{F}}{\partial \param }\right]_{ij}
&= \left.\frac{\partial f}{\partial \theta_j}\right\rvert_{(\steady_i^{\text{mfa}}, \param)}+
\left.\frac{\partial g}{\partial x_j} 
\right\rvert_{(\steady_i^{\text{mfa}},\steady_{\text{eff}},\param)}
\frac{\partial \steady_{\text{eff}}}{\partial \theta_j} \nonumber\\
&\quad+\left. \frac{\partial g}{\partial \theta_j}\right\rvert_{(\steady_i^{\text{mfa}}, \steady_{\text{mfa}},\param)}.
\end{align}
Applying the Implicit Function Theorem, we obtain
\begin{align}
\frac{\partial \steady_{\text{eff}}}{\partial \theta_j}&= -
\frac{  {\partial \mathbf{G}}/{\partial \theta_j} }{  {\partial \mathbf{G}}/{\partial \steady_{\text{eff}} } }.
\end{align}

\section{Block-regular Networks\label{sup:brn}}
We analyze convergence properties of the NLS framework on a simplified structure: Block-regular network.  
\begin{definition}
\textup{Block-regular Network (BRN).} A block-regular network, denoted as $\mathcal{B}$, consists of $M$ disconnected blocks $\mathcal{B}_{1},\ldots,\mathcal{B}_{M}$. Each block $\mathcal{B}_{i}$ has  $N_{i}$ nodes sampled from a regular network with degree $r_{i}$ and $r_{i}\neq r_{j}$ for $i\neq j$. $r_{i}$ is called the block degree of $\mathcal{B}_{i}$.
\end{definition}

We establish the empirical identifiability of the specified ground truth parameter $\param^*$ (defined in main text)  given noiseless data for ecological and gene regulatory dynamics by finding a block degree sequence $r_1,\ldots,r_M$ s.t. within a subset of $\mathbb{R}^{d}$, the surfaces 
\begin{align}
\{
\param\in \mathbb{R}^{d}: \ 
f(x_i,\param) + r_i g(x_i, x_i,\param) = 0\}_{i=1}^{M} 
\end{align}
only intersect at $\param^*$, where $x_i$ are the given fixed steady states.

Suppose we observe more than five blocks $\mathcal{B}_{i}$ for ecology dynamics. The state $x_i$ of block $\mathcal{B}_{i}$ is a root of 
\begin{align}
B + x\left(1-\frac{x}{K}\right)\left(\frac{x}{C}-1\right) + r_{i} \frac{x^2}{D+E'x}=0\label{eq:mfa-eco}
\end{align}
where $E'=E+H$.
Using the first three blocks, we have 
  \begin{align}
  \label{lin-system1}
  \begin{bmatrix}
  1& x_{1}^{2} & -x_{1}^{3} \\
  1& x_{2}^{2}& -x_{2}^{3}\\
  1& x_{3}^{2}& -x_{3 }^{3}
  \end{bmatrix}\begin{bmatrix}
  B\\K'\\C'
  \end{bmatrix} &= 
  \begin{bmatrix}
  C_{1}(D,E')\\C_{2}(D,E')\\C_{3}(D,E
  ')
  \end{bmatrix}
  \end{align}
  where  $C_{i}(D,E')=\frac{ r_{i} x_{i}^{2} }{ D+ E x_{i} }+x_{i}$, $K'=\frac{1}{K}+\frac{1}{C}$ and $C'=1/KC$.
  Assuming the matrix in the LHS of Eq.~(\ref{lin-system1}) is invertible, we can express $B,K',C'$ uniquely using any fixed $(D,E')$.
  Using the forth block, we obtain 
  \begin{align*}
  \psi(D,E', r_{1},\ldots,r_{4})&=
  a_{1}\frac{1}{D+E'x_{1}} 
  + a_{2}\frac{1}{D+E'x_{2}}\\ &
  + a_{3}\frac{1}{D+E' x_{3}} 
  + a_{4} \frac{1}{D+ E' x_{4}} +a_{5} \\&= 0
  \end{align*}
  for some constants $a_{i}$.
  The first five blocks yield the following set of  $D,E'$ that agree with the steady state data.
  \begin{align*}
  \Gamma_{[i, i + 4]} 
  \coloneqq \bigcap_{j=i+3,i+4} &\{(D,E')\in [0,5]^2\ |\\ &\psi(D,E', r_{i},r_{i+1},r_{i+2},r_{j})=0 \}.
  \end{align*}
  In ecology networks, $K\geq C$~\cite{gao2016universal}.
  $K,C$ are the roots to the quadratic equation $  x^{2} +  K'C'  x +1/C'=0$. Therefore $(K,C)$ is uniquely determined given $K',C'$.
  It suffices to find a block degree sequence that pins down a unique $(D,E')$ and the corresponding determinant of the matrix in (\ref{lin-system1}) is invertible.
  Set the degree sequence as $\boldsymbol{\delta}=[7,\ 13,\ 17,\ 23,\ 29,\ 58,\ 68,\ 79,\ 89,\ 100]$. We evaluate the steady states from an initial condition of 6 for each block, with ground truth parameter $[B,K,C,D,E']=[0.1,5,1,5,1]$.
  The determinant of the matrix in system (\ref{lin-system1})
  is $-1434.95\neq 0$. We plot $(D,E')$ that are in $\Gamma_{[1,5]}$ and $\Gamma_{[6,10]}$ in Fig.~\ref{fig:brn-uniqueness}(a). We treat values with magnitude smaller than or equal to 1e-7 as zero to account for numerical imprecision and found $\Gamma_{[1,5]}\cap \Gamma_{[6,10]}=\{(5,1)\}$.

\begin{figure}[h]
    \centering
    \includegraphics[width=0.6\textwidth]{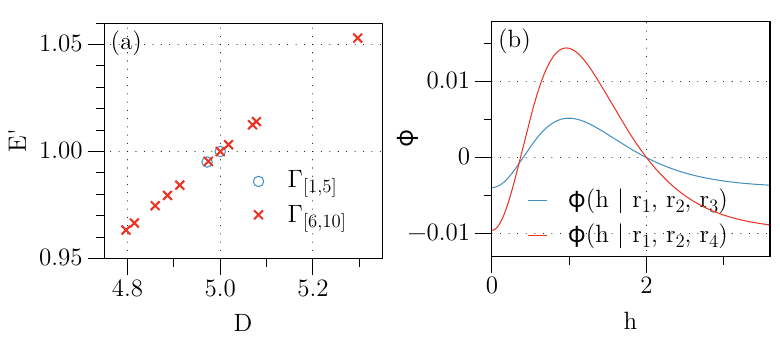}
    \caption{(a) The blue (red) dots live in the subspace defined by degree sequence $[7,13,17,23,29]$ ($[58,68,79,89,100]$). The two subspaces of ecological parameter space $(D,E')\in [0,5]^2$ intersect at $(D,E')=(5,1)$, which uniquely determines the other parameters $(B,K,C)$.  (b) The blue (red) curve's intersection with $y=0$ denotes the $h$ that agree with the steady-state equations defined by degree sequence $[7,13,17]$ ($[7,13,23]$). The two subspaces of gene regulatory parameter space intersect at $h = 2$, which uniquely determines the other parameters $f=B=1$.\label{fig:brn-uniqueness}}
  \end{figure}
 
  We now consider the gene regulatory dynamics.
  Suppose we observe a BRN with more than four blocks.
  We have four equations of the following form 
  \begin{align*}
  -B x_{i}^{f}+r_{i} \frac{x_{i}^{h}}{x_{i}^{h}+1}&=0.
  \end{align*}
  Let $C_{i}(h)= r_{i}\frac{x_{i}^{h}}{x_{i}^{h}+1}$.
  Blocks $\mathcal{B}_{1},\mathcal{B}_{2}$ give 
  \begin{align}
  B&= C_{1}(h)x_{1}^{-f}\label{eq:B}\\
  f&= \frac{\ln \left(C_{2}(h)/C_{1}(h)\right)}{\ln \left(x_{2}/x_{1}\right)}\label{eq:f}
  \end{align}
  Substituting $B,f$'s expressions into the equation for $\mathcal{B}_{3}$, we eliminate $B,f$ and obtain 
  \begin{align*}
  \label{eq:reg-unique}
   &\ln \frac{x_{3}^{h}+1}{x_{1}^{h}+1}
   - \frac{\ln\left(x_{3}/x_{1}\right)}{\ln\left(x_{2}/x_{1}\right)} \ln \frac{x_{2}^{h}+1}{x_{1}^{h}+1}\\
   &-\ln \frac{r_{3}}{r_{1}}+ \frac{\ln\left(x_{3}/x_{1}\right)}{\ln\left(x_{2}/x_{1}\right)} \ln \frac{r_{2}}{r_{1}}=0
  \end{align*}
  The LHS is a function of $h$, dependent on the first three blocks.
  Denote the LHS as $\phi(h\ |\  r_{1},r_{2},r_{3})$.
  Using block $\mathcal{B}_{4}$, we have $\phi(h\ |\  r_{1},r_{2},r_{4}) =0$.
  Simulate a BRN with block degrees $[7,13,17,23]$ and ground truth parameters $[1,1,2]$. 
  We observe that $\{h\in [0,2]: \phi(h\in [0,2]\ |\  r_{1},r_{2},r_{3})=0\}\cap \{h:\phi(h\ |\  r_{1},r_{2},r_{4})=0\}=\{2\}$, which uniquely determines $f=B=1$ using equations  Eq.~(\ref{eq:B}), Eq.~(\ref{eq:f}) (Fig.~\ref{fig:brn-uniqueness}(b)).

{\bf Connection to General Networks.}
We generalize the identifiability of the ground truth parameter \math{\param^*} for a BRN to a general network.
One can show that if the state deviation objective for a BRN has a unique minimizer \math{\param^*}, then \math{\param^*} is the unique parameter that produces the noiseless steady states of a slightly perturbed BRN. 

\begin{theorem}\label{thm:perturb_brn}
Let $\mathcal{B}$ be an arbitrary BRN such that each block is a regular network. Randomly perturb $\mathcal{B}$ $k$ times to obtain $\mathcal{G}$. Let $\param^*$ be the ground truth parameter of $\mathcal{G}$. Denote the steady state of network $\mathcal{G}$ under $\param^*$ as $\mathbf{x}^* \coloneqq \Steady(\param^*, \mathcal{G})$. $\mathbf{x}' \coloneqq \Steady(\param^*, \mathcal{B})$ is the steady state of the BRN $\mathcal{B}$. Given the following assumptions
\begin{inparaenum}[(i)] 
    \item The perturbation $k$ is small.
    \item $\Steady(\param,\mathcal{G})$ is continuous w.r.t. $\param$ around $\param^*$,
    \item $\lVert\Steady(\param, \mathcal{B})- \mathbf{x}'\rVert\gg 0$ for $\param$ that is far from $\param^*$,
\end{inparaenum}
we conclude $\param^*$ is the unique parameter that yields the steady states  $\mathbf{x}^*$ of $\mathcal{G}$.
\end{theorem}
\begin{proof} (Sketch)
Suppose there exists a parameter $\boldsymbol{\theta}$ such that $\boldsymbol{\theta}\neq \boldsymbol{\theta}^*$ and $\Steady(\boldsymbol{\theta}, \mathcal{G}) = \mathbf{x}^*$.
Case 1: $\boldsymbol{\theta}$ is close to $\boldsymbol{\theta}^*$. By continuity of the function $\Steady(\cdot,\cdot)$ around $\boldsymbol{\theta}^*$ (Assumption $\textit{2}$), we eliminate this possibility.
Case 2: $\boldsymbol{\theta}$ is far from $\boldsymbol{\theta}^*$. Due to the little-perturbation assumption (Assumption $\textit{1}$), $\mathcal{B}$ and $\mathcal{G}$ have a similar structure and degree distribution. Then
\begin{align*}
\Steady(\boldsymbol{\theta}, \mathcal{B}) &\approx \Steady(\boldsymbol{\theta}, \mathcal{G}) \quad \text{(by similarity of $\mathcal{G}$ and $\mathcal{B}$)}\\
&= \mathbf{x}^* \quad \text{(by assumption)}\\
&= \Steady(\boldsymbol{\theta}^*, \mathcal{G})  \quad \text{(by definition)}\\
&\approx \Steady(\boldsymbol{\theta}^*, \mathcal{B}) \quad \text{(by similarity of $\mathcal{G}$ and $\mathcal{B}$)}\\
&= \mathbf{x}' \quad \text{(by definition)}
\end{align*}
$\mathbf{x}^* \approx \mathbf{x}'$ violates the assumption that the steady state of $\mathcal{B}$ should be distinguishable given parameters that are far away from each other (Assumption $\textit{3}$).
It follows that $ \lVert\Steady(\param,\mathcal{G}) -\mathbf{x}^* \rVert^2 $ has a unique minimum.
\end{proof}

\bibliography{reference}%
\end{document}